\documentclass[11pt]{article}

\usepackage{geometry}
\usepackage[T1]{fontenc}
\usepackage{amsfonts}
\usepackage{amsmath}
\usepackage{amssymb}
\usepackage{amsthm}
\usepackage{bbm}
\usepackage{bm}
\usepackage{csquotes}
\usepackage{mathrsfs}
\usepackage{verbatim}
\usepackage{setspace}
\usepackage{color}
\usepackage{pdfsync}
\usepackage{enumitem}
\usepackage{graphicx}
\usepackage{subcaption}
\usepackage{tikz}
\usetikzlibrary{patterns}
\usepackage{placeins}
\usepackage[normalem]{ulem}

\theoremstyle{plain}
\newtheorem{theorem}{Theorem}[section]
\newtheorem{proposition}[theorem]{Proposition}
\newtheorem{lemma}[theorem]{Lemma}
\newtheorem{corollary}[theorem]{Corollary}

\theoremstyle{definition}

\newtheorem{remark}[theorem]{Remark}

\theoremstyle{remark}

\usepackage{graphicx,accents}

\renewenvironment{thebibliography}[1]{%
\begin{oldthebibliography}{#1}%
\setlength{\baselineskip}{1em}
\small
}%
{%
\end{oldthebibliography}%
}
\newcommand{\eps}{\varepsilon}

\newcommand{\E}{\mathbb{E}}

\newcommand{\R}{\mathbb{R}}

\newcommand{\cA}{\mathcal{A}}

\newcommand{\cC}{\mathcal{C}}

\newcommand{\cF}{\mathcal{F}}

\newcommand{\cL}{\mathcal{L}}

\newcommand{\cQ}{\mathcal{Q}}

\newcommand{\bA}{\mathbf{A}}
\newcommand{\bB}{\mathbf{B}}
\newcommand{\bC}{\mathbf{C}}

\newcommand{\bP}{\mathbf{P}}
\newcommand{\bM}{\mathbf{M}}
\newcommand{\bN}{\mathbf{N}}

\newcommand{\1}{\mathbf{1}}

\newcommand{\teps}{\tilde{\eps}}

\newcommand{\qandq}{\quad\mbox{and}\quad}
\newcommand{\qqandqq}{\qquad\mbox{and}\qquad}

\newcommand{\mykill}[1]{}
\numberwithin{equation}{section}

\usepackage[pdfborder={0 0 0}]{hyperref}
\hypersetup{
  urlcolor = black,
  pdfauthor = {Marcel Nutz, Kevin Webster, Long Zhao},
  pdfkeywords = {Internalization; Market Making; Central Risk Book; Optimal Execution; Price Impact},
  pdftitle = {Unwinding Stochastic Order Flow: When to Warehouse Trades},
  pdfsubject = {Unwinding Stochastic Order Flow: When to Warehouse Trades},
  pdfpagemode = UseNone
}

\begin{document}

\title{Unwinding Stochastic Order Flow: When to Warehouse Trades}
\date{\today}
\author{
  Marcel Nutz%
  \thanks{
  Depts.\ of Statistics and Mathematics, Columbia University, mnutz@columbia.edu. Research supported by NSF Grants DMS-1812661, DMS-2106056, DMS-2407074.}
  \and
  Kevin Webster%
  \thanks{Dept.\ of Mathematics, Columbia University.}
  \and
  Long Zhao%
  \thanks{Dept.\ of Statistics, Columbia University, long.zhao@columbia.edu.}
  }
  
\maketitle \vspace{-1.2em}

\begin{abstract}
We study how to unwind stochastic order flow with minimal transaction costs. Stochastic order flow arises, e.g., in the central risk book (CRB), a centralized trading desk that aggregates order flows within a financial institution. The desk can warehouse in-flow orders, ideally netting them against subsequent opposite orders (internalization), or route them to the market (externalization) and incur costs related to price impact and bid-ask spread. We model and solve this problem for a general class of in-flow processes, enabling us to study in detail how in-flow characteristics affect optimal strategy and core trading metrics. Our model allows for an analytic solution in semi-closed form and is readily implementable numerically. Compared with a standard execution problem where the order size is known upfront, the unwind strategy exhibits an additive adjustment for projected future in-flows. Its sign depends on the autocorrelation of orders; only truth-telling (martingale) flow is unwound myopically. In addition to analytic results, we present extensive simulations for different use cases and regimes, and introduce new metrics of practical interest.
\end{abstract}

\vspace{.3em}

{\small


}
\vspace{.6em}

\maketitle

\section{Introduction}

This paper aims to model and solve the unwind problem for stochastic order flow; that is, an optimal execution-type problem where instead of a fixed position known at the initial time, the trading desk receives a dynamic flow of orders to be unwound.

Our work is motivated by the emergence of central risk books (CRBs) in most investment banks and several large trading firms over the last few years. 
The CRB is a trading desk aiming to minimize transaction costs by aggregating internal order flows. In an ideal situation, all other units within the organization (and possibly the asset class) interact exclusively with the CRB, instead of trading directly in the market.
In \cite[p.\,2]{CRBefinancialcareers}, CRBs are defined as ``the places that (in theory) centralize the execution of trades for the whole bank.'' As the total volume resulting from such centralization is very large,\footnote{While public data seem to be unavailable, an equity CRB at a large institution would process a nontrivial percentage of the total trading volume in US equities, currently about \$500 billion daily \cite{CBOEvolume}.} 
transaction costs are of paramount importance.
Orders across different business lines and strategies can have idiosyncratic directions and timing, leading to a stochastic order flow (\emph{in-flow} orders from the desk's perspective). For instance, an options desk's delta hedging orders are likely independent of prime brokerage clients' block trades. The CRB nets opposite in-flows against one another, a process named \emph{internalization}. Remaining orders can either be routed to the market (\emph{out-flow} orders), a process named \emph{externalization}, or \emph{warehoused}, meaning that the CRB temporarily keeps the position on its book. Out-flow orders incur transaction costs including transient and permanent price impact. The most favorable scenario for a warehoused order is to be netted against an opposite order later on, thus avoiding transaction costs. Even disregarding that  possibility, 
warehousing can be beneficial to reduce costs; for instance, immediately externalizing a sizable block order would often lead to an unnecessarily large price impact.

Our model is also relevant to market makers. Similarly to banks, market makers face orders from clients ranging from retail traders to institutional funds, resulting in stochastic in-flow.\footnote{Some market makers also use the name CRB. For instance, CEO D.\ Cifu of global market maker Virtu emphasized internalization opportunities at a Goldman Sachs Financial Services Conference \cite{Cifu}:  ``we're the only firm in the world that has a central risk book of retail order flow, prop flow, and now institutional agency customers. [\,\dots] So there's a real opportunity to be a real internalizer between the buy side and the retail side.''} The paper~\cite{CarteaJaimungalJia.20} discusses a model where an ambiguity-averse dealer trades in a currency triplet to liquidate a large position while also receiving order flow. The distinction between in-flow and out-flow is most evident in foreign exchange (FX), where clients send orders directly to dealers, and dealers unwind their remaining inventory on a central venue (e.g., EBS), also called the lit market. 
The Bank of England's Fair and Effective Markets Review \cite[p.\,11]{bofe2014} stressed the importance of internalization in FX: ``Market participants have indicated that dealers with large enough market share now internalise up to 90\% of their client orders.''\footnote{The report further highlighted that interdealer trading as a share of overall FX market turnover declined by one third since the late 1990s, and attributed this to the growth of internalization \cite[p.\,59]{bofe2014}. For reference, the Bank for International Settlements \cite{bis2016} estimated that interdealer trading in FX accounted for 1 trillion USD \emph{daily.}} The market maker must decide which orders to warehouse and which to unwind, an aspect emphasized recently by~\cite{Cartea2022}. Also recently, \cite{BergaultDrissiGueant.22} suggested market makers leverage the same mechanisms as CRBs to unwind sizable inventory considering price impact. 
The present paper does not discuss how the market maker should set quotes, usually considered the core problem in the literature on market making (see \cite{CarteaBook} and the references therein). Our model is relevant to market makers specifically in the regime where the inventory gets too large to revert passively through limit orders. Then, market makers use dedicated algorithms to reduce the inventory aggressively. As the inventory simultaneously varies due to in-flow orders, this constitutes an unwind problem for a stochastic flow.

Our study focuses on the unwind problem for stochastic order flow. Compared with the extensive work on optimal execution, two essential differences arise. First, the final volume to be executed is unknown at the initial time. Instead, the in-flow of orders is a stochastic process and the unwind strategy must take into account projected future orders. Second, there are opportunities for orders to net, which introduces internalization benefits not directly modeled in the standard setting.
Compared with the literature on market making, our in-flow is uncontrolled. Instead of studying quotes, we focus on minimizing the transaction costs of unwinding a potentially large inventory subject to stochastic shocks. 

The study closest to ours, in terms of methodology and focus on transient price impact, is the concurrent work \cite{MuhleKarbeOomen.23} on pre-hedging. In their model, a dealer learns about a potential trade through a client's request for quote. The probability of winning the trade depends on the dealer's quote, and there is uncertainty about the timing of the potential transaction. The authors determine the optimal pre-hedge strategy: the dealer starts to build inventory ahead of the potential trade, taking into account transaction costs and the expected size of the trade. Clearly, the potential trade plays a role similar to our stochastic in-flow. On the other hand, as there is only a single trade, internalization does not occur in their model. 
At the opposite extreme, \cite{ButzOomen.19} consider internalization in a queueing model. The dealer can skew their prices to encourage inventory-reducing customer flow, but there is no option to externalize. While this model is quite different from ours, \cite{ButzOomen.19} gives insightful comments and data sources about internalization in the FX market. The more recent work \cite{BarzykinBergaultGueant.22, BarzykinBergaultGueant.23} combines internalization and externalization. Here a dealer sets quotes to attract flow and simultaneously hedges in a separate liquidity pool. The authors highlight the optimal behavior: pure internalization when inventory is small, aggressive externalization when inventory is large. The liquidity pool in this model bears instantaneous (and permanent impact) cost, but not the transient impact cost that is the main driving force in our study. A related equilibrium model is formulated in \cite{BankEkrenMuhleKarbe.21}. This model features clients (liquidity takers) and dealers who absorb their demands. Dealers trade competitively in the interdealer market and also in an open market subject to instantaneous transaction costs (but no impact cost). The study solves for the equilibrium price and highlights the endogenous price impact incurred by the clients. An early contribution with stochastic inventory shocks is \cite{CarteaJaimungal.15} which considers optimal execution with market and limit orders. Limit orders are executed at random times, leading to stochastic shocks in the inventory. Another paper close to ours in terms of application is \cite{Cartea2022}, where a foreign exchange market maker streams bespoke quotes to an informed and a noise trader, and also trades in the lit market. The market maker learns the signal of the informed trader from his trades and uses it to decide on quotes and externalization of her order flow, and on speculative trades. In particular, the model derives an unwind strategy for a specific stochastic order flow that is endogenous to the equilibrium. 

By contrast, our approach is to model the stochastic order flow exogenously in a reduced but general form. That will allow us to analyze the trading metrics of the corresponding optimal strategy from an input--output perspective and tackle some of the key trading questions. We do not model incentives the desk can provide their clients to induce orders or the information the desk can glean from the clients' orders. In particular, the desk we model does not engage in speculative trades considering client information.

This paper is organized as follows. In the remainder of the Introduction, Section~\ref{se:synopsis} summarizes a selection of the high-level results, while Section~\ref{se:background} reviews background and modeling considerations around transaction costs and order flow. Section~\ref{se:theory} details the mathematical model, then analytically derives the solution and some of its qualitative properties, with most proofs being reported in Appendices~\ref{se:proofs} and~\ref{se:maxPrinciple}. Section~\ref{se:numRes} defines a number of trading metrics and presents simulation studies on how they react to autocorrelation and volatility of the in-flow. Further numerical results are reported in Appendix~\ref{se:appendixNumerical}:  relation to classical optimal execution; sensitivities to the spread cost parameter, the martingale driving the in-flow process, and the initial impact state; and finally an empirical study on autocorrelation of orders on the public trading tape. Appendix~\ref{se:proofs} gathers the proofs based on the dynamic programming approach while Appendix~\ref{se:maxPrinciple} proceeds through the stochastic maximum principle.

\subsection{Synopsis}\label{se:synopsis}

We propose and solve a flexible model for unwinding an exogenously given order flow, rich enough to serve as a practical order-scheduling algorithm.
In particular, it allows for order flows with momentum or reversal, time-varying volatility and block trades. Moreover, it takes into account impact and spread costs with possibly time-varying liquidity parameters.

We model price impact by a generalized Obizhaeva--Wang model while instantaneous (spread) costs during limit order book (LOB) trading  are modeled in reduced form by a quadratic penalty on trading speed. The desk thus trades continuously during the day, but also places block orders in the opening and closing auctions which only carry impact cost. The desk's inventory must be flat  after the close. See also Section~\ref{se:background} for further modeling background and terminology.

It turns out that our model can be transformed into a linear-quadratic stochastic control problem, which is key to its tractability. The optimal intraday trading speed $q_{t}$ takes the form 
\[
q_{t} =  f_{t} X_{t} + g_{t} Y_{t} + h_{t} Z_{t}
\]
where the time-varying coefficients $f_{t},g_{t},h_{t}$ are explicitly determined through an ordinary differential equation (ODE), and even provided in closed form for the case of constant liquidity parameters. The coefficients describe the desk's reaction to the current inventory $X_{t}$, the impact state~$Y_{t}$, and the cumulative in-flow~$Z_{t}$. We will show analytically that $h_{t}=0$ when there is no in-flow after the initial position, meaning that $q_{t}=f_{t} X_{t} + g_{t} Y_{t}$ would be the solution in a deterministic optimal-execution framework. The additional term $h_{t} Z_{t}$ captures forecasted future in-flows. The sign of $h_{t}$ depends on whether the in-flow exhibits momentum or reversal. For momentum, the adjustment leads to faster unwinding as more same-sided orders are expected; this result supports the ``overtrading'' of orders often seen in practice. The adjustment is zero if the in-flow is a martingale, thus justifying traders' intuition that ``truth-telling'' flows can be executed myopically. The tractability of the model allows us to rigorously study various other properties by analytic means (see Section~\ref{se:theory}). This is useful not only to understand the general dependencies of the output, but also for a desk to recognize, e.g., if its algorithm incurs excessive model risk by chasing an arbitrage implicit in its liquidity estimates.

Transaction costs such as impact costs are difficult to study empirically because the unaffected price (that would have prevailed without the transaction) cannot be observed. Instead, key metrics are inherently model-dependent. Having a consistent, computable model for general in-flows opens the door to analyzing cost metrics side-by-side with model-independent quantities like the internalization rate. The most essential trading question is how unwind strategy and cost metrics depend on the characteristics of the in-flow; namely, autocorrelation and volatility. In our numerical simulations, we observe that all core trading metrics strongly depend on these characteristics, in line with the industry mantra, \emph{know your client.} Some key insights are:

\begin{itemize}
\item Expected trading costs (per order notional) are minimized at a particular in-flow volatility. Below that threshold, additional volatility increases the internalization rate, whereas above the threshold, internalization plateaus and further volatility necessitates more aggressive trading, driving up costs. As a take-away, a desk may benefit from incentivizing additional order flow to increase internalization up to the threshold. 

\item Autocorrelation strongly affects all core trading metrics. Momentum requires more aggressive trading, increasing both impact costs and spread costs. The ratio between those costs is approximately stable. Reversion leads to more internalization and more warehousing until the close; the closing trade represents a larger fraction of the total trading volume.

\item Our misspecification analysis shows that it is preferable to err on the \emph{optimistic} side when estimating autocorrelation: overestimating momentum sharply increases costs by way of overly aggressive trading and missed netting opportunities, whereas trading too slowly only incurs a moderate additional cost by trading too much on the close. Misspecifying in-flow volatility bears no transaction cost since the unwind strategy's feedback form is independent of volatility; only the calculation of expected costs is affected.
\end{itemize} 
Section~\ref{se:numRes} and Appendix~\ref{se:appendixNumerical} present numerous further simulation studies, and also introduce new metrics of interest to practitioners.

Numerical experiments highlight that our model encompasses a wide spectrum of regimes, with scenarios resembling classical optimal execution and market making occurring as extreme cases. The regime depends on the realization of the random in-flow trajectory. To illustrate this, Figure~\ref{fig:ACSamplePathIntro} shows two particular realizations from the same (martingale) in-flow process and the model's corresponding optimal unwind strategy.
\begin{figure}[bth]
    \centering
    \includegraphics[width=1.0\textwidth]{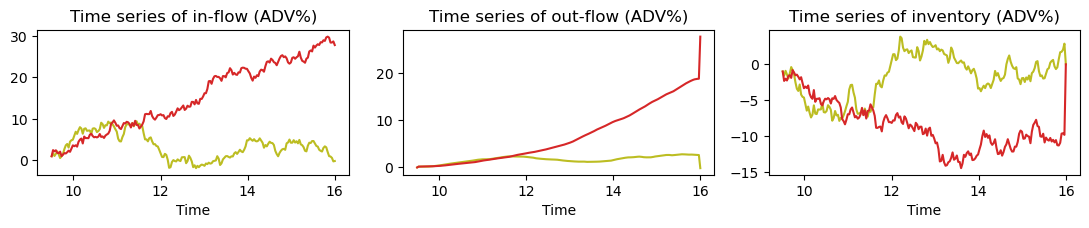}
    \caption{Two realizations of the same model illustrating extreme regimes.}
    \label{fig:ACSamplePathIntro}
\end{figure}%
One realization (red) is closer to classical optimal execution: Most orders have the same sign and the strategy unwinds aggressively to reduce the inventory. Internalization occurs and leads to a smoothing between in and out-flow, but is not sufficient to control the inventory. Transaction costs are dominated by price impact. Such scenarios are frequent when the initial position is large compared to in-flow volatility, and when in-flow exhibits momentum. The other realization (yellow) resembles a favorable market-making scenario: orders approximately cancel, the unwind strategy trades slowly most of the time, and the inventory never gets large. Spread costs are not negligible compared to impact costs. Such scenarios are frequent when in-flow volatility is large compared to the initial position, and when in-flow exhibits reversion.
The relevance of models that can capture all regimes was previously emphasized by \cite[p.\,37]{ButzOomen.19}:
\begin{displayquote}
``An important message is that the polarisation between internalisers and externalisers is overly simplistic and that instead one needs to view liquidity providers across a continuum with passive internalisers at one end [\,\dots] and eventually externalisers at the other end.''
\end{displayquote}

\subsection{Background}\label{se:background}

Next, we discuss our model's main ingredients. The trading metrics for its analysis are discussed in Section~\ref{se:numSetup}.

\paragraph{In-flow} or \emph{order flow}\footnote{The term ``order flow'' is used differently in~\cite{CarteaJaimungal.16}. There, order flow represents the orders sent to the market by other market participants and the focus is on how those orders impact prices available during execution.} refers to the incoming orders from the desk's clients. %
Practitioners propose stochastic processes for in-flow based on internal historical data. Depending on the organization, the clients, and the order options offered to them, the desk will observe different characteristics. The in-flow is generally stochastic, potentially with jumps if block orders are allowed, and there is a pronounced time-of-day effect: most orders arrive shortly after the open or before the close, as reflected in market volumes. For examples of intraday trading patterns, see \cite{CarteaBook, BouchaudBook}. In-flow may exhibit momentum, for instance, if clients slice their orders to mask the total volume, or if they target momentum strategies. 
See \cite{Toth2015} for an in-depth empirical analysis of market flow autocorrelation. 
Because splitting of meta orders is common in LOB trading, order flow on the public trading tape exhibits momentum, a finding that we confirm empirically for all S\&P 500 stocks (see Section~\ref{se:empiricalAutocorr}). We emphasize that the situation might be quite different for a CRB with internal clients, say, which ideally (though maybe not in practice) deliver ``truth-telling'' (martingale) flow by being upfront about their volume.
Less frequently, in-flow may exhibit mean reversion, for instance, if clients follow a mean-reversion strategy or if they tend to cancel orders mid-execution.\footnote{The latter can occur if an execution desk allows for reductions of order size during execution but not for increases, as clients may then deliberately overstate the initial order and leverage the free option to reduce the order's size.} Our model uses a stochastic process 
$$
  dZ_{t}=-\theta_{t}Z_{t}\,dt + \sigma_{t}\,dW_{t}
$$
with time-varying parameters to allow for a wide range of client scenarios. It exhibits momentum (reversion) if $\theta_{t}<0$ ($>0$). In practice, execution desks have ample historical data to determine suitable parameters and want to ensure their in-flow model reflects their client base. Indeed, different desks achieve different model parameters, with more sophisticated desks achieving milder order flow. For instance, a CRB with strong relationships may attract clients with less momentum. Another example is a low-latency market-maker: they could cancel limit orders and widen spreads when predicting more toxic (momentum) flow, effectively achieving a milder in-flow compared to the average market participant. We refer to the follow-up paper \cite{BarzykinBoyceNeuman.24} for a model where $\theta_{t}$ is unknown and filtered from the observed in-flow.

We use a Brownian motion driver $\sigma_{t}\,dW_{t}$ to keep the model readily understandable to a broader audience but show in Section \ref{se:jumps} that the results extend to a general martingale driver~$dM_{t}$ with minor modifications, thus covering in-flows with jumps, self-excitement, etc. In the numerical experiments of Section~\ref{se:numRes}, we use a driver of finite-variation so that trading metrics can be expressed ``per order notional.''

\paragraph{Unwind trades} or \emph{out-flow} are the market-facing trades of the desk. We denote by~$Q_{t}$ the cumulative unwind trades up to time~$t$; this is the control variable in our model.  %
An order scheduling algorithm typically uses block trades at the open and closing auctions and trades continuously on the limit order book during the day.
The desk's available actions may differ from the clients' options for submitting orders. For example, clients might submit block orders intraday, leading to jumps in a CRB's position that cannot be unwound immediately (at a reasonable cost). Another example of a client-triggered block trade is when a market maker's resting order in a dark pool gets filled by a large mutual fund.

The desk's \emph{outstanding position} or \emph{inventory} $X_{t}=Q_{t}- Z_{t}$ is the difference between the unwind trades and the in-flow. For simplicity, our model focuses on a single trading day and a desk that does not hold overnight positions. Hence, we impose the liquidation constraint $X_{T}=0$ at the time horizon~$T$ corresponding to the end of the closing auction.

\paragraph{Transaction costs} come in two forms: price impact and bid-ask spread. For CRBs servicing major financial institutions, trading costs are primarily driven by price impact, but spread costs are not negligible. For example, Nasdaq's \emph{The 2022 Intern's Guide to Trading} \cite{Mackintosh2022} gives numbers for large-cap stocks of 30 basis points for price impact and 5 basis points for spread costs. 

A price impact model generally consists of two elements. First, a push factor, also called Kyle's lambda, describes a market fill's immediate effect on the midprice. Second, a decay kernel describes how quickly this dislocation reverts in the absence of further fills. The interplay between push and reversion dictates the optimal trading speed. Empirically, traders observe a time-of-day effect for Kyle's lambda: it is largest near the open and smallest near the close. See, e.g., \cite{Cont2014} and \cite{JMK2022b}. See \cite{Salek2023} for estimates of price impact in auctions. 

The observed execution price $P_{t}$ is thought of as the sum $P_{t}=S_{t}+Y_{t}$ of an ``unaffected'' price~$S_{t}$ (that would have prevailed without our trading) and a market impact process $Y_{t}$. Our model uses a generalized Obizhaeva--Wang (OW) model~\cite{obizhaeva2013optimal}  with time-dependent coefficients,
\begin{equation*}
   dY_t = -\beta_{t} Y_t \,dt + \lambda_{t}\,dQ_{t}.
\end{equation*}
Thus a trade of size $\delta Q_{t}$ immediately dislocates the price by $\lambda_{t}\delta Q_{t}$, and the dislocation decays exponentially at rate $\beta_{t}$. This applies to block trades as well as to continuous trading. To the best of our knowledge, the generalized Obizhaeva--Wang model was first introduced in \cite{FruthEtAl.2013}. We use time-dependent but deterministic coefficients. (In practice, these curves are typically fixed at the beginning of the day, based on data from previous trading sessions, especially the projected intra-day trading volume in the market. Thus, the coefficients are deterministic intra-day.) A stochastic resilience $\beta_t$ was considered in~\cite{GraeweHorst.17}. Both $\beta_t$ and $\lambda_t$ are stochastic in the formulation of~\cite{AckermannKruseUrusov.22} which in addition allows for a stochastic terminal condition. In principle, this allows to model stochastic in-flows, though the paper does not elaborate on that application. Other price impact models consider non-exponential kernels, as empirically studied in~\cite{Bouchaud2004} and solved by~\cite{AbiJaber2022}, or non-linear price impact, as introduced by \cite{Bouchaud2004, Alfonsi2010} and studied by \cite{Gatheral2011b, JMK2022b, Hey2023}.

We do not model permanent impact (i.e., impact without decay). Because initial and terminal inventory are fixed, permanent impact only adds a constant to the cost and does not affect the optimal strategy, because any round trip trade has zero permanent impact cost (see \cite[Proposition~3.3]{FruthEtAl.2013} for a detailed proof).

For costs related to the bid-ask spread, we use a straightforward model as in~\cite{almgren.chriss.01}: intraday trading incurs a quadratic instantaneous cost $\frac12\eps_{t} q_{t}^{2}$ on the trading speed $q_{t}=\dot{Q}_{t}$. While this forces intraday trading to be continuous, no such cost is charged for the auctions. We abuse the term ``spread costs'' in that this reduced-form cost stands in not just for the spread costs of market orders, but also for adverse selection incurred by limit orders and any other instantaneous costs that do not change the price in a persistent way.

If the desk nets orders of opposite sign instead of routing them to the market, it saves the associated spread costs entirely. We emphasize that the situation is different for impact costs, as this is often overlooked: Suppose a buy order pushes up the price. It incurs an impact cost, but a sell order immediately following it benefits from the increased price, so that the sum of the costs is approximately zero. In general, the cost saving due to netting is related to the impact decay that would have occurred between the trades.

To focus our study on the core internalization and externalization trade-off, and to provide results that are agnostic to the proprietary signals the desk has access to, we model the unaffected price $S$ as a martingale. This is consistent with our assumption that the in-flow is driven by a martingale---because order flow leads to price changes via price impact, an exogenous drift term in the order flow amounts to a signal for the price. For considerations around proprietary trading and toxic flows, see \cite{Cartea2022,LehalleNeuman.19,NeumanVoss.22}.

\section{Analytic Results: Optimal Strategy and Cost}\label{se:theory}

This section contains our theoretical results.
Section~\ref{se:problemFormulation} details the mathematical formulation of the problem. Section~\ref{se:optStrat} provides the optimal unwind strategy and its expected cost, for the case of in-flow driven by Brownian motion and possibly time-varying liquidity parameters. Section~\ref{se:details} presents the high-level steps of its derivation and further details; the technical proofs are reported in Appendix~\ref{se:proofs}. Section~\ref{se:jumps} generalizes the results to in-flows driven by a more general martingale (possibly with jumps, non-Markovian, etc.). On the other hand,  Section~\ref{se:cstLiq} specializes to constant liquidity parameters and provides a closed-form solution for that case. The latter result is derived by a different approach than the rest (maximum principle instead of dynamic programming); the proof is reported in Appendix~\ref{se:maxPrinciple}.

\subsection{Problem Formulation}\label{se:problemFormulation}

We fix the time horizon $T > 0$ and bounded, measurable functions $\beta, \lambda, \eps: [0,T] \to (0, \infty)$, $\theta:[0, T] \to \R$ and $\sigma:[0, T] \to [0, \infty)$, where $\eps,\lambda$ are also bounded away from zero. Moreover, $\lambda$ is differentiable and its derivative $\dot\lambda$ is bounded. 
For convenience, we denote $\gamma_{t} := \log \lambda_{t}$. We also introduce an additional value $\lambda_{0-} \in (0, \infty)$ representing the opening auction's liquidity parameter.
We assume throughout that 
\begin{equation}\label{eq:noarbCondLambda}
  \lambda_{0-} \leq \lambda_{0} \qqandqq 2\beta_{t} + \dot\gamma_{t}>0,
\end{equation}
which is trivially satisfied if the liquidity parameters are constant in time. We will see in Section~\ref{se:details} that these are no-arbitrage conditions, excluding profitable roundtrip trades and ensuring that our problem is convex. See \cite{Goyal2022} for estimates of price impact based on TAQ data  and a comparison across continuous trading, closing auction, and opening auction.

\begin{remark}[Price manipulation]\label{rk:priceManipulation}
Fast increases in liquidity can give rise to arbitrage, or \emph{price manipulation,} as first noted by~\cite{HubermanStanzl.04}: if one order pumps the price and a subsequent order of opposite sign has less price impact, a profitable roundtrip arises, unless the price has sufficiently reverted between the trades. We will see in Section~\ref{se:details} that~\eqref{eq:noarbCondLambda} exclude such profitable (or free) roundtrips.
For one of our comparative statics below (Proposition~\ref{pr:properties}), we will require a strengthening of the second part of~\eqref{eq:noarbCondLambda}, namely that $\beta_{t} + \dot\gamma_{t}>0$. This is related to additionally excluding ``transaction-triggered price manipulation,'' first studied by \cite{AlfonsiSchiedSlynko.12, GatheralSchiedSlynko.12}. See Remark~\ref{rk:transactionTriggered} for further details, and Remark~\ref{rk:practicalPriceManipulationCond} for practical aspects. Note that the price impact parameter for the closing auction is $\lambda_{T}=\lim_{t\to T}\lambda_{t}$. A trader may be interested in specifying a smaller value for Kyle's lambda in the auction compared to LOB trading before the close. However, this would immediately give rise to price manipulation strategies. That may, of course, be an artifact of our simplified model for the auction---an area left for further investigation.  (One could specify a value larger than $\lambda_{T}$, but this is not relevant in practice; see \cite{Goyal2022}.)
\end{remark}

Let $(\Omega, \cF, (\cF_{t})_{t \in [0, T]}, \mathbb{P})$ be a filtered probability space satisfying the usual conditions, endowed with a Brownian motion $W = (W_{t})_{t \in [0, T]}$ and a c\`adl\`ag martingale $S = (S_{t})_{t\in[0,T]}$ which models the unaffected price. %
The process~$S$ is required to satisfy a minor integrability condition (see Footnote~\ref{fot:qAdmissible}).

We define the in-flow process $Z$ by $Z_{0} = z \in \R$ and 
\begin{align}
    dZ_{t} = -\theta_{t} Z_{t} \, dt + \sigma_{t} \, dW_{t}, \quad t \in [0, T].
\end{align} 
This process represents the aggregated order flow the desk faces. The sign convention follows the client's perspective; e.g., a positive increment corresponds to a buy order, causing a negative increment in the desk's outstanding position. The initial value~$z$ represents the orders present at the beginning of our problem.

Next, consider a c\`adl\`ag semimartingale $(Q_{t})_{t \in [0, T]}$ with $Q_{0-} = 0$, representing the desk's cumulative unwind trades (out-flow). Here and below, certain processes are defined at time $t=0-$ to allow for a jump at the initial time, in this case, a block trade $\Delta Q_{0}=Q_{0}-Q_{0-}=Q_{0}$ in the opening auction. Given~$Q$, its price impact process $Y$ is defined by $Y_{0-} = y$ and
\begin{align} \label{eq:impact}
    d Y_{t} = -\beta_{t} Y_{t} \, dt + \lambda_{t-} \, dQ_{t}, \quad t \in [0, T].
\end{align}
Here and below, our convention is that integrals such as $\int_{0}^{t}\lambda_{s-} \, dQ_{s}$ are taken over the \emph{closed} interval $[0,t]$, as for instance in~\cite{Protter2005}. In particular, a jump $\Delta Q_{0}$ in the integrator leads to a jump of the integral at $t=0$. 
In our model, $Q$ will typically have block trades at the initial and terminal times, so that~\eqref{eq:impact} gives rise to jumps in the impact process,
\begin{align} \label{eq:impact jumps}
    \Delta Y_{0}=Y_{0} - y = \lambda_{0-} Q_{0} \qquad \text{and} \qquad \Delta Y_{T}=Y_{T} - Y_{T-} = \lambda_{T} (Q_{T} - Q_{T-}).
\end{align}
Thus, $y$ represents the impact state when our problem begins (e.g., impact from the previous day's trading) and $\lambda_{0-} Q_{0}$ is the change in impact state due to the block trade in the opening auction.

In our model, the trading speed bears a quadratic cost on~$(0,T)$. Hence unwind strategies $(Q_{t})$ trade in an absolutely continuous fashion on~$(0,T)$, in addition to the block trades at $t=0$ and $t=T$:
\begin{align} \label{eq:Q main}
    Q_{t} := J_{0} + Q^{c}_{t} + J_{T} \1_{t = T}, \qquad Q^{c}_{t}:= \int_{0}^{t} q_{s} \, ds.
\end{align}

Specifically, we define an \emph{admissible strategy} as a triplet $(J_{0}, q, J_{T})$. Here $J_{0}\in\R$ and $J_{T}\in L^{2}(\cF_{T})$. Moreover, $(q_{t})_{t\in[0,T]}$ is a progressively measurable process such that $\E[\int_{0}^{T} q^{2}_{t} \, dt] < \infty$ and such that the local martingale $\int_0^t Q^{c}_s\,dS_s$ is a true martingale. Finally, \begin{equation}\label{eq:admConstraint}
  J_{0} + \int_{0}^{T} q_{s} \, ds + J_{T} = Z_{T}.
\end{equation}

Thus $J_{0}$ and $J_{T}$ are the sizes of the block trades $\Delta Q_{0}$ and $\Delta Q_{T}$, and $q$ is the speed of absolutely continuous trading during regular hours. Finally, \eqref{eq:admConstraint} is the liquidation constraint $Q_{T}=Z_{T}$. Occasionally it will be convenient to switch between seeing the strategy as a triplet $(J_{0}, q, J_{T})$ or a process~$Q$. We then write $Q\equiv (J_{0}, q, J_{T})$ to indicate that the relationship~\eqref{eq:Q main} holds.

Let~$\cA$ denote the set of admissible strategies. Given~$Q\equiv (J_{0}, q, J_{T}) \in \cA$, the execution price is
\begin{align}\label{eq:execPrice}
    P_{t} := 
    \begin{cases}
        S_{t} + \frac{1}{2}(Y_{t-} + Y_{t}), & t \in \{0, T\}, \\
        S_{t} + Y_{t} + \frac{1}{2} \eps_{t} q_{t}, & t \in (0, T).
    \end{cases}
\end{align}
The returns from trading at the execution price $P_{t}$ can be broken down into three parts.
\begin{enumerate}
\item The frictionless wealth
\begin{equation*}
S_0 J_0 + \int_0^T S_t q_t \, dt + S_T J_T
\end{equation*}
reflects returns from the unaffected price (not caused by the unwind strategy).
\item The \emph{impact cost}
\begin{equation}
\frac{1}{2}\left(Y_{0-} + Y_0\right) J_0 + \int_0^T Y_t q_t \,dt + \frac{1}{2}\left(Y_{T-} + Y_T\right) J_T\label{eq:impactCostC0}
\end{equation}
reflects returns caused by the the price impact of the unwind strategy.
\item The \emph{spread cost}  (or more precisely, \emph{instantaneous cost})
\begin{equation}
\frac{1}{2}\int_0^T \eps_t q^2_t \,dt.\label{eq:spreadCostC0}
\end{equation}
\end{enumerate}
As described in Section~\ref{se:background}, $t \in \{0, T\}$ corresponds to the opening and closing auctions. There is no spread cost (meaning that block trades are possible), but there is price impact. The first part of~\eqref{eq:execPrice} (and the jump terms in~\eqref{eq:impactCostC0}) reflect that the Obizhaeva--Wang model prices a block trade at the average price over the volume of the block. For $t \in (0, T)$, we also have the instantaneous cost, and as trading is necessarily continuous, we can write $Y_{t}$ instead of $\frac{1}{2}(Y_{t-} + Y_{t})$ in the second part of~\eqref{eq:execPrice} (and in the integral of~\eqref{eq:impactCostC0}).

The expected execution cost of a given unwind strategy is
\begin{align}\label{eq:execCost}
    \cC(J_{0}, q, J_{T}) := \E \left[ \int_{0}^{T} P_{t} \, dQ_{t} \right] = \E \left[ P_{0} J_{0} + \int_{0}^{T} P_{t} q_{t} \, dt + P_{T} J_{T} \right].
\end{align}
The desk's aim is to minimize this cost,
\begin{align}\label{mainProblem}
    \inf_{(J_{0}, q, J_{T}) \in \cA} \cC(J_{0}, q, J_{T}).
\end{align}
We remark that the possibility of the closing block trade in our model is important for the existence of a strategy satisfying the liquidation constraint. Without a block trade, that constraint can in general only be satisfied if the in-flow volatility $\sigma_{t}$ tends to zero as $t\to T$; see \cite{bank.al.17} for a precise statement.

\subsection{Optimal Strategy}\label{se:optStrat}

The next result provides a semi-closed form solution for the optimal strategy and optimal cost of~\eqref{mainProblem}; i.e., explicit up to solving an ordinary differential equation (ODE). The ODE does not depend on the realization of the in-flow path, so that it can pre-computed in an implementation. For the case of constant liquidity parameters $\beta,\lambda,\eps$, Section~\ref{se:cstLiq} will even provide a fully explicit solution; on the other hand, Section~\ref{se:jumps} will cover more general in-flow processes. We emphasize that the implementation of the below formulas is straightforward (even if some look lengthy).

\begin{theorem}\label{th:main}
The problem~\eqref{mainProblem} has a unique solution. The optimal strategy $(J_{0}, q, J_{T}) \in \cA$ is
    \begin{align}
    q_{t} &=  f_{t} X_{t} + g_{t} Y^{c}_{t} + h_{t} Z_{t}, \quad t \in [0, T], \label{eq:qMain}\\[.4em]
    J_{0} &=  r^{-1}(g_{0-} + \eta_{0-}) y + r^{-1}(-f_{0-} + h_{0-}) z, \nonumber\\
    J_{T} &= Z_{T} - J_{0} - \int_{0}^{T} q_{t}\,dt, \nonumber
    \end{align}
    where $\eta_{0-} = -\eps^{-1}_{0} \left( 1 - {\lambda_{0-}}/{\lambda_{0}} \right) \leq 0$ and $r= -f_{0-} - \lambda_{0-} \left(g_{0-} + \eta_{0-} \right) > 0$ are constants and
    \begin{align}\label{eq:fghFormulas}
    f_{t}&= - \eps^{-1}_{t}(A_{t} + \lambda_{t} B_{t}), \nonumber\\ 
    g_{t}&= - \eps^{-1}_{t}(B_{t} + \lambda_{t} C_{t}), \\
    h_{t}&= - \eps^{-1}_{t}(D_{t} + \lambda_{t} E_{t}) \nonumber
    \end{align}
    for $t \in \{0-\}\cup[0, T]$.
    Here, $A_{t}, B_{t}, \dots$ are deterministic functions given by the Riccati ODE system~\eqref{eq:riccati} below, and $\eps_{0-}:=\eps_{0}$, $A_{0-}:=A_{0}$, $B_{0-}:=B_{0}$, $\dots$. Finally, the state processes in~\eqref{eq:qMain} are
    \begin{align} \label{eq:forward Yc}
    \begin{cases}
        dX_{t} = q_{t} \, dt - dZ_{t}, &\quad X_{0} = J_{0} - z \\
        dY^{c}_{t} = (-\beta_{t} Y^{c}_{t} + \lambda_{t} q_{t}) \, dt, &\quad Y^{c}_{0} = y + \lambda_{0-} J_{0} \\
        dZ_{t} = -\theta_{t} Z_{t} \, dt + \sigma_{t} \, dW_{t}, &\quad Z_{0} = z.
    \end{cases}
    \end{align}
    The optimal expected cost~\eqref{mainProblem} is
    \begin{align*}
        \cC(J_{0}, q, J_{T})
        =  v(0, J_{0} - z, y + \lambda_{0-} J_{0}, z) 
        + \frac{1}{2} \Big\{\left( {\lambda^{-1}_{0-}} - {\lambda^{-1}_{0}} \right) (y + \lambda_{0-} J_{0})^{2} 
        - \lambda^{-1}_{0-} y^{2}
        + \E[S_{T}Z_{T}] \Big\}
    \end{align*}
   where $v$ is the quadratic polynomial~\eqref{eq:ansatz} with coefficients $A_{0}, B_{0}, \dots$.
\end{theorem}

The key formula~\eqref{eq:qMain} states that the optimal trading speed~$q_{t}$ is a linear combination of the auxiliary state processes $(X_{t},Y^{c}_{t},Z_{t})$ with time-varying coefficients $(f_{t},g_{t},h_{t})$ determined through the Riccati ODE. The coefficients are discussed in more detail in the next proposition. The processes $(X_{t},Y^{c}_{t})$  in~\eqref{eq:forward Yc} are the inventory and impact state, but only up to the close. More precisely, they are equal to $(Q_{t}-Z_{t},Y_{t})$ on $[0,T)$ but are continuous at $t=T$, so that $(X_{T},Y^{c}_{T})=(Q_{T-} - Z_{T-},Y_{T-})$ corresponds to the values immediately before the closing auction. Moreover, the jumps from the opening auction have been absorbed into the initial positions $(X_{0},Y^{c}_{0})$. As before, $Z$ is the in-flow.

The optimal trading speed~$q_{t}$ in \eqref{eq:qMain} is given by the feedback function $f_{t} x + g_{t} y + h_{t} z$; thus $f_{t}$ reflects the reaction to outstanding position, $g_{t}$ to impact state, and $h_{t}$ to in-flow. The next proposition summarizes their key properties. (See Remark~\ref{rk:transactionTriggered} for a discussion of the less important regime $\beta_{t} + \dot\gamma_{t}\leq 0$ related to transaction-triggered price manipulation.)

\begin{proposition}\label{pr:properties}

\begin{enumerate}
\item Let $\beta_{t} + \dot\gamma_{t}>0$. We have $f_{t}<0$ for all $t<T$: a larger outstanding position causes faster selling.

\item We have $g_{t}<0$ for all $t<T$: higher impact state causes slower selling.

\item We have $h\equiv0$ if $\theta\equiv0$ (martingale in-flow). Let $\beta_{t} + \dot\gamma_{t}>0$. Then $h$ is monotone decreasing wrt.\ $\theta$; in particular, $h\geq0$ if $\theta\leq0$ (in-flow with momentum) and $h\leq0$ if $\theta\geq0$ (in-flow with reversion).

\item The functions $f$ and $g$ depend on the liquidity parameters $(\eps,\beta,\lambda)$, but not on the characteristics $(\theta,\sigma)$ of the in-flow. The function $h$ additionally depends on $\theta$. The volatility $\sigma$ does not affect the strategy (in feedback form) but increases the execution cost $\cC(J_{0}, q, J_{T})$.
\end{enumerate}  
\end{proposition}

The feedback function quantifies some of our intuition on internalizing or externalizing orders. As seen in~(i) and~(ii), $f_t$ promotes externalization when the inventory becomes sizable and $g_{t}$ promotes internalization when the impact state is high (in absolute value). 
In view of~(iii) and~(iv), the contribution of order flow predictions becomes clear: in comparison with a standard optimal execution problem where the total order is given at $t=0$, the unwind speed is adjusted to the random in-flow $Z$ by adding the term $h_{t}Z_{t}$. If the in-flow has no drift ($\theta\equiv0$), then no adjustment is necessary. Traders say in that case that the order flow is \emph{truth-telling}; it does not predict future orders.  Of course, the realization of $q_{t}$ still depends on the realization of $Z_{t}$, but the optimal strategy is myopic in the sense that it does not anticipate \emph{future} changes of the in-flow (because these changes have mean zero).
By contrast, the strategy does adjust for projected future changes when there is momentum, as expected if clients are splitting meta orders, or when there is reversion (see also Section~\ref{se:background}). The coefficient $h_t$ promotes internalization if order flow is reverting, whereas in the more common case of momentum flow, $h_t$ leads to additional externalization as one predicts more same-sided orders.

\subsection{Reduction to Linear-Quadratic Control}\label{se:details}

To solve our problem, we first derive the optimal strategy and cost \emph{after} an exogenously given block trade of arbitrary size $J_{0}$ in the opening auction. In particular, for $J_{0}=0$, we solve the problem in a market without opening auction (or, the problem started at a time after the opening auction). In a second step, we will then optimize the result of this problem over $J_{0}\in\R$ (see Proposition~\ref{prop:optJ0}) to solve the full problem~\eqref{mainProblem}.

To derive the optimal continuous strategy, we recast our problem: the liquidation constraint can be reformulated to say that the block trade $J_{T}$ in the closing auction must match whatever is the outstanding position immediately before that auction. We can thus focus on the continuous trading on $(0,T)$ and let $J_{T}$ be determined via the constraint~\eqref{eq:admConstraint}. This effectively replaces the state constraint with a terminal cost: after the continuous trading, the desk is charged the cost of the block trade determined via the constraint. A priori, this terminal cost is intractable, but the next result shows that it can be transformed into a rather convenient form. More precisely, our problem is expressed as a tractable control problem of linear-quadratic (LQ) type.

\begin{proposition}\label{pr:costReformulated}
    Let $(J_{0}, q, J_{T})\in\cA$. The execution cost~\eqref{eq:execCost} can be stated as 
    \begin{align}\label{eq:costReformulated}
        \cC(J_{0}, q, J_{T}) = \frac{1}{2} \E \bigg\{  \int_{0}^{T} & \left[ \frac{2\beta_{t} + \dot\gamma_{t}}{\lambda_{t}} (Y^{c}_{t})^{2} + \eps_{t} q^{2}_{t} \right] dt + 
        \frac{1}{\lambda_{T}}(\lambda_{T}X_{T} - Y^{c}_{T})^{2} \nonumber \\
        & + \left(\frac{1}{\lambda_{0-}} - \frac{1}{\lambda_{0}} \right) (Y^{c}_{0})^{2} - \frac{y^{2}}{\lambda_{0-}} + S_{T}Z_{T} \bigg\},
    \end{align}
    where $(X_{t}, Y^{c}_{t}, Z_{t})_{t\in[0,T]}$ are defined by~\eqref{eq:forward Yc}.
\end{proposition}

The proof is stated in Appendix~\ref{se:proofs}.
We note that the cost~\eqref{eq:costReformulated} no longer contains $J_{T}$. As $J_{0}$ is considered fixed at the moment, the last three terms in the cost~\eqref{eq:costReformulated} are independent of~$q$ and can thus be neglected for the purpose of finding the optimal continuous strategy $q$. (When $\lambda_{0}=\lambda_{0-}$, those  terms are independent of~$J_{0}$ as well, but that is not crucial.) On a more technical note, we observe that the no-arbitrage conditions~\eqref{eq:noarbCondLambda} are necessary to make~\eqref{eq:costReformulated} (strictly) convex. In particular, we have the following.

\begin{corollary}[No free roundtrips]\label{co:noFreeRoundtrips}
   Suppose that there is no in-flow ($Z\equiv0$) and that the initial impact state is $y=0$. If $(J_{0},q,J_{T})\in\cA$ satisfies $\cC(J_{0}, q, J_{T})=0$, then $J_{0}=J_{T}=0$ and $q=0$ $dt$-a.e. 
\end{corollary} 

\begin{proof}
  Recall that $2\beta_{t} + \dot\gamma_{t}>0$ by~\eqref{eq:noarbCondLambda}. If $J_{0}\neq0$, then $Y^{c}_{0} = y + \lambda_{0-} J_{0}\neq0$. As $t\mapsto Y^{c}_{t}$ is continuous, \eqref{eq:costReformulated} shows that $\cC(J_{0}, q, J_{T})>0$. Similarly, if $q\neq0$, then $\cC(J_{0}, q, J_{T})>0$ as $\eps>0$. Finally, if $J_{0}=0$ and $q=0$, the liquidation constraint implies $J_{T}=0$.
\end{proof} 

In summary, we have transformed our problem into a more standard LQ control problem with control~$q$. The next result solves this problem using dynamic programming. As usual in the dynamic programming approach, we consider the problem from arbitrary initial states $(x,y,z)$ and initial time~$t$. Moreover we ignore the terms $(\lambda^{-1}_{0-} - \lambda^{-1}_{0}) (Y^{c}_{0})^{2}  - \lambda^{-1}_{0-} y^{2}  + S_{T}Z_{T}$ in the cost~\eqref{eq:costReformulated} as their expectation is the same for any admissible choice of~$q$.

\begin{proposition} \label{pr:auxControl}
    Fix $t \in [0,T)$ and let $\cQ$ be the set of progressively measurable processes~$q$ with $\E[\int_{t}^{T} q^{2}_{s} \, ds]<\infty$. 
For $q\in\cQ$ and $x,y,z\in\R$, define the state processes $(X_{s}, Y^{c}_{s}, Z_{s})_{s\in[t,T]}$ by
    \begin{align} \label{eq:forward}
    \begin{cases}
        dX_{s} = q_{s} \, ds - dZ_{s}, &\quad X_{t} = x \\
        dY^{c}_{s} = (-\beta_{s} Y^{c}_{s} + \lambda_{s} q_{s}) \, ds, &\quad Y^{c}_{t} = y \\
        dZ_{s} = -\theta_{s} Z_{s} \, ds + \sigma_{s} \, dW_{s}, &\quad Z_{t} = z
    \end{cases}
    \end{align}
    and the value function $v(t, x, y, z)$ by
    \begin{align}\label{valueFunDef}
     \inf_{q \in \cQ} \frac{1}{2} \E\left\{  \int_t^T \left[ \frac{2\beta_{s} + \dot\gamma_{s}}{\lambda_{s}} (Y^{c}_{s})^{2} + \eps_{s} q^{2}_{s} \right] ds 
     +\frac{1}{\lambda_{T}}(\lambda_{T}X_{T} - Y^{c}_{T})^{2} \right\}.
\end{align}
    Then    
    \begin{align} \label{eq:ansatz}
    v(t,x, y, z) &= \frac{1}{2} A_{t} x^{2} + B_{t} x y + \frac{1}{2} C_{t} y^{2} + D_{t} x z + E_{t} y z + \frac{1}{2} F_{t} z^{2} + K_{t},
    \end{align}
    where $A_{t}, B_{t},\dots,K_{t}$ are deterministic functions defined in Proposition~\ref{pr:riccati} below. Moreover, the unique optimal control for~\eqref{valueFunDef} is given by the feedback function
    \begin{align} \label{eq:opt alpha}
    q^{*}(s,x,y,z) =  f_{s} x + g_{s} y + h_{s} z,
    \end{align}
    where 
    \begin{align*}
    f_{s}&:= - \eps^{-1}_{s}(A_{s} + \lambda_{s} B_{s}), \\
    g_{s}&:= - \eps^{-1}_{s}(B_{s} + \lambda_{s} C_{s}), \\ 
    h_{s}&:= - \eps^{-1}_{s}(D_{s} + \lambda_{s} E_{s}).
    \end{align*}
\end{proposition}

The proof is stated in Appendix~\ref{se:proofs}. The following functions where used in~\eqref{eq:ansatz}.

\begin{proposition}\label{pr:riccati}
The following Riccati ODE system has a unique solution 
on $[0,T]$:
\begin{align} \label{eq:riccati}
    & \begin{cases}
        \dot A_{t} = \eps^{-1}_{t} (A_{t} + \lambda_{t} B_{t})^2, & A_{T} = \lambda_{T} \\ %
        \dot B_{t} = \eps^{-1}_{t} (A_{t} + \lambda_{t} B_{t})(B_{t} + \lambda_{t} C_{t}) + \beta_{t} B_{t}, & B_{T} = -1 \\
        \dot C_{t} = \eps^{-1}_{t} (B_{t} + \lambda_{t} C_{t})^{2} + 2\beta_{t} C_{t} - \lambda^{-1}_{t}(2 \beta_{t} + \dot\gamma_{t}), & C_{T} = \lambda^{-1}_{T} \vspace{0.3cm} \\ 
        \dot D_{t} = \eps^{-1}_{t} (A_{t} + \lambda_{t} B_{t})(D_{t} + \lambda_{t} E_{t}) - \theta_{t} (A_{t} - D_{t}), & D_{T} = 0 \\
        \dot E_{t} = \eps^{-1}_{t} (B_{t} + \lambda_{t} C_{t})(D_{t} + \lambda_{t} E_{t}) - \theta_{t} (B_{t} - E_{t}) + \beta_{t} E_{t}, & E_{T} = 0 \\
        \dot F_{t} = \eps^{-1}_{t} (D_{t} + \lambda_{t} E_{t})^{2} - 2\theta_{t} (D_{t} - F_{t}), & F_{T} = 0 \vspace{0.3cm} \\
        \dot K_{t} = - \sigma^{2}_{t} (A_{t} - 2 D_{t} + F_{t}) / 2, & K_{T} = 0.
    \end{cases}
\end{align}
\end{proposition}

Again, the proof is stated in Appendix~\ref{se:proofs}.
Some observations regarding~\eqref{eq:riccati} are in order: 
\begin{enumerate}
    \item $A,B,C$ form an autonomous system depending on the market impact parameters $\eps_{t},\beta_{t},\lambda_{t}$ but not on the characteristics of the in-flow~$Z$.  
    \item $D,E,F$ also depend on the mean reversion $\theta_{t}$. If the in-flow is a martingale ($\theta\equiv0$), then $D, E, F \equiv 0$.

    \item $K$ also depends on the volatility of~$Z$. If the in-flow is deterministic  ($\sigma\equiv0$), then $K \equiv 0$.
\end{enumerate}

The following analytic result on $A_{t}, B_{t},\dots,K_{t}$ implies the properties of the optimal unwind speed reported in Proposition~\ref{pr:properties} by way of the formulas for $f_{t},g_{t},h_{t}$ in~\eqref{eq:fghFormulas}.

\begin{proposition} \label{pr:riccatiAdditionalProperties}
Let $A_{t}, B_{t},\dots,K_{t}$ be the functions defined in~\eqref{eq:riccati}. 
\begin{enumerate}
\item The matrix 
\begin{align*}
    \begin{bmatrix} A_{t} & B_{t} \\ B_{t} & C_{t} \end{bmatrix}
\end{align*}
is nonnegative definite, $0 \leq C_{t} \leq \lambda^{-1}_{t}$, and $B_{t} + \lambda_{t} C_{t} > 0$ for $t \in [0, T)$.

\item Let $\beta_{t} + \dot\gamma_{t}>0$. Then
\begin{align*}
    A_{t} + \lambda_{t} B_{t} > \max(D_{t} + \lambda_{t} E_{t}, 0) \qandq B_{t} < \min( E_{t}, 0), \qquad t \in [0, T).
\end{align*}

\item Let $\beta_{t} + \dot\gamma_{t}>0$. Let $\tilde \theta:[0,T]\to\R$ be measurable and bounded, and define $\tilde D,\tilde E$ like $D,E$ but with~$\tilde\theta$ instead of~$\theta$. If $\theta_{t} \geq \tilde \theta_{t}$ for $t\in [0, T]$, then
\begin{align*}
  D_{t} + \lambda_{t} E_{t} \geq \tilde D_{t} + \lambda_{t} \tilde E_{t} \qandq E_{t} \leq \tilde E_{t}, \qquad t \in [0, T].
\end{align*}
In particular, $\theta_{t} \geq 0, \ t \in [0, T]$ implies $D_{t} + \lambda_{t} E_{t} \geq 0 $ and $E_{t} \leq 0$.
\end{enumerate}
\end{proposition}

The proof is stated in Appendix~\ref{se:proofs}. The work \cite{GraeweHorst.17} on stochastic resilience (i.e., stochastic $\beta_{t}$) features a backward stochastic differential system with singular terminal condition that plays a similar role as our (much simpler) ODE system~\eqref{eq:riccati}. We have benefited from several ideas in~\cite{GraeweHorst.17} for proving Proposition~\ref{pr:riccatiAdditionalProperties}.

\begin{remark}[Transaction-triggered price manipulation]\label{rk:transactionTriggered}
Under the condition $\beta_{t} + \dot\gamma_{t} > 0$, Proposition~\ref{pr:riccatiAdditionalProperties} yields in particular that $f_{t}<0$ for $t \in [0, T)$, which is the result reported in Proposition~\ref{pr:properties}\,(i). Then, a larger positive outstanding position leads to faster selling and a larger negative position leads to faster buying. If $\beta_{t} + \dot\gamma_{t} < 0$, the proof of Proposition~\ref{pr:riccatiAdditionalProperties} shows that $f_{t}>0$ for $t\in[t_{0},T)$, for some $t_{0}<T$. Thus, the opposite behavior occurs: given a large positive outstanding position, the optimal strategy will initially add to the position, only to unwind later. This corresponds to the transaction-triggered price manipulation discussed in detail in \cite[Theorem~8.4 and Corollary~8.5]{FruthEtAl.2013}. While the weaker condition $2\beta + \dot \gamma \geq 0$ from~\eqref{eq:noarbCondLambda} rules out profitable round trips from zero initial position, it does not rule out profitable round trips from non-zero initial position.

If this seems unintuitive, considering the boundary case where $\beta_{t} + \dot{\gamma}_{t} = 0$ for $t\in[0,T]$ may be helpful. In this regime, the benefit of trading earlier to exploit the reversion~$\beta_{t}$ of the impact process exactly offsets the benefit of trading later to exploit the increase in liquidity (decay of $\lambda_{t}$) over time. Specifically, the proof of Proposition~\ref{pr:riccatiAdditionalProperties} shows that $f_{t} = h_{t} = 0$, whereas $g_{t} < 0$ for $t \in [0,T)$. Thus, the optimal trading speed~\eqref{eq:opt alpha} only depends on the current impact state, not on the current outstanding position. If the impact state at~$t$ is zero, no trading occurs over $[t, T)$; if the impact state at $t$ is positive (negative), the optimal policy sells (buys) over $[t, T)$, to exploit the reversion of the impact process. In either case, the remaining outstanding position is unwound in the closing auction.
\end{remark}

\begin{remark}[Practical use of no-price manipulation conditions]\label{rk:practicalPriceManipulationCond}
Due to our framework's inherent tractability, practitioners can directly apply no-price manipulation conditions on live data. Indeed, the condition only involves the liquidity curves $\beta, \lambda$, which the trading team typically builds from volume predictions or order book data. For instance, see \cite{Cont2014} and \cite{JMK2022b}, which fit parametric models of $\lambda$ as functions of order book depth and market volumes, respectively.

Intuitively, %
the condition states that
``liquidity cannot increase faster than impact reverts without triggering price manipulation opportunities.''
If a given liquidity forecast breaches a no-price manipulation condition, traders can take two actions:
\begin{enumerate}
\item Switch to a simpler externalization strategy, e.g., VWAP or a constant parameter OW model. This avoids having the algorithm engage in odd trading patterns due to potential liquidity mirages.
\item Investigate the liquidity surge that triggered the breach: Is there a data or model problem, or is the sudden liquidity caused by abnormal market behavior, e.g., order spoofing?
\end{enumerate}

In that sense, no-price manipulation conditions are to trading algorithms as no-arbitrage conditions are to option pricing models. Models trade under the assumption they don't exist, but traders monitor the market for unexplained behavior.
\end{remark}

\begin{remark}[Cost on Inventory]\label{rk:riskAversion}
  In the literature on optimal execution, a standard way of including ``risk aversion'' is to add a running cost of the form $\int_{0}^{T} \alpha  X_{t}^2 dt$ in the objective function. This form is tractable as it is compatible with the linear-quadratic framework. Recall however that~$X_{t}$ denotes the number of shares held, rather than their value, so that the interpretation as a risk term is not immediate. In~\cite{almgren.chriss.01}, $\alpha=\tilde\alpha \tilde{\sigma}^{2}$ where $\tilde{\sigma}$ is the asset's instantaneous volatility and $\tilde\alpha\geq0$ is a parameter to be chosen. Then, $\alpha  X_{t}^2$ is the instantaneous squared volatility of the value of the portfolio, hence indeed a measure of risk. In any event, a running cost with $\alpha>0$ incentivizes faster unwinding. %
  
  More generally, let $\alpha: [0,T] \to [0, \infty)$ be a bounded measurable function. We describe how our formulas change when a running cost $\int_{0}^{T} \alpha_{t}  X_{t}^2 dt$ is added to the objective. In fact, the conclusion of Proposition~\ref{pr:auxControl} remains unchanged, except that the definition of the functions $A_{t}, B_{t},\dots$ changes. Namely, in Proposition~\ref{pr:riccati}, the equation for $A$ becomes 
$\dot A_{t} = \eps^{-1}_{t} (A_{t} + \lambda_{t} B_{t})^2 - 2\alpha_{t}$. This is seen by inspecting the proof of Proposition~\ref{pr:auxControl}. As a consequence, the functions $A_{t}, B_{t},\dots$, and then also the coefficients $f_{t},g_{t},h_{t}$ of the strategy, depend on $\alpha_{t}$. It remains true that the in-flow volatility $\sigma_{t}$ only affects $K_{t}$, and hence does not affect the feedback form of the strategy.

\end{remark} 

Next, we return to the derivation of the optimal strategy and cost. 
We can now deduce the optimal expected cost when the block trade in the opening auction is prescribed to be of size $j_{0}\in\R$, by evaluating the auxiliary value function~\eqref{eq:ansatz} at the corresponding initial conditions and adding back the constant terms of~\eqref{eq:costReformulated}. This yields the following.

\begin{corollary} \label{cor:fixedJ0}
    Given $j_{0} \in \R$, the optimal execution cost is %
    \begin{align}
        \inf_{\substack{(J_{0}, q, J_{T}) \in \cA \\ J_{0} = j_{0}}} \, & \cC(J_{0}, q, J_{T})
        =  v(0, j_{0} - z, y + \lambda_{0-} j_{0}, z) \nonumber \\
        &+ \frac{1}{2} \Big\{\left( {\lambda^{-1}_{0-}} - {\lambda^{-1}_{0}} \right) (y + \lambda_{0-} j_{0})^{2} 
        - \lambda^{-1}_{0-} y^{2}
        + \E[S_{T}Z_{T}] \Big\}. \label{eq:optCostJ0}
    \end{align}
    The optimal trading speed in feedback form is given by~\eqref{eq:opt alpha}, with the state processes started at time~$0$ with the initial condition $(x,y,z)=(j_{0}-z,y+\lambda_{0-}j_{0},z)$.
\end{corollary}

Finally, we determine the optimal size $J_{0}$ for the block trade in the opening auction.

\begin{proposition} \label{prop:optJ0}
    The optimal initial block trade is 
    \begin{align}\label{eq:optJ0}
        J_{0} &=  r^{-1}(g_{0-} + \eta_{0-}) y + r^{-1}(-f_{0-} + h_{0-}) z,
    \end{align}
    where $\eta_{0-} := -\eps^{-1}_{0} (1 - {\lambda_{0-}}/{\lambda_{0}}) \leq 0$, 
    \begin{align*}
     g_{0-} + \eta_{0-} < 0 \qandq  r:= -f_{0-} - \lambda_{0-} \left(g_{0-} + \eta_{0-} \right) > 0.
    \end{align*}
    If $\beta_{t} + \dot \gamma_{t} > 0$ for $t \in [0, T]$, then $-f_{0-} + h_{0-} > 0$.
\end{proposition}

The proof is stated in Appendix~\ref{se:proofs}.
We see that the coefficients of~$y$ and~$z$ in~\eqref{eq:optJ0} resemble the ones of the trading speed~$q$ and have similar comparative statics (cf.\ Proposition~\ref{pr:properties}). %
That the coefficient of~$z$ contains~$f$ and not only~$h$, is merely a consequence of our notation: the initial in-flow $z$ also represents an outstanding position of $x:=-z$ at the beginning of the auction, and now $-f_{0-}z=f_{0-}x$ is analogous to the formula of~$q$.

\subsection{In-Flows with Jumps}\label{se:jumps}

In this section we generalize our results to more general in-flow processes
    \begin{equation}\label{eq:ZgenMart}
      dZ_{s} = -\theta_{s} Z_{s} \, ds + dM_{s}, \qquad Z_{0} = z,
    \end{equation}
where $M_{t}$ is a general martingale instead of $\int_{0}^{t} \sigma_{s}\, dW_{s}$. Specifically, allowing a process with jumps is important as a CRB's in-flow may contain large block orders, and market makers may receive large block trades from dark pools. Self-exciting (hence non-Markovian) dynamics are also relevant. The main take-away is that the formula for the optimal strategy remains unchanged (similarly as, e.g., in \cite{garleanu.pedersen.16}), whereas the expression of the optimal expected cost involves an integral of the quadratic variation of~$M$. This generalizes our previous observations in Proposition~\ref{pr:properties}, where the optimal strategy was independent of~$\sigma$ while the cost was an integral against $\sigma^{2}_{s}\,ds$ (as seen in the formula for $K$ in~\eqref{eq:riccati}).

Let $M$ be a c\`adl\`ag, square-integrable martingale with $M_{0}=0$.
 We  consider a variant of the control problem in Proposition~\ref{pr:auxControl}, now only for the initial time $t=0$ (since $M$ was not assumed to be Markovian). To state its solution, we define the deterministic functions $A_{t},\dots,E_{t}$ as in~\eqref{eq:riccati}, whereas $K_{t}$ is generalized to the stochastic process
\begin{equation}\label{eq:KgenMart}
      K_{t} := \E\bigg[ \int_{t}^{T}  \frac12 (A_{s} - 2 D_{s} + F_{s}) \,d[M,M]_{s} \,\bigg|\, \cF_{t}\bigg].
\end{equation}
At a technical level, we also need to amend the definition of an admissible strategy and impose a mild integrability condition (see the proof).

\begin{proposition}\label{pr:genMart}
  Define the state processes $X,Y^{c}$ as in~\eqref{eq:forward} while~$Z$ is generalized to~\eqref{eq:ZgenMart}. Moreover, define $A,\dots,E$ as in Proposition~\ref{pr:riccati} whereas $K$ is generalized to~\eqref{eq:KgenMart}.
  \begin{enumerate}
  \item The problem~\eqref{valueFunDef} admits a unique optimal strategy $q^{*}$, given by the same feedback function as in~\eqref{eq:opt alpha}.        
  \item The value of~\eqref{valueFunDef} is given by
    \begin{align*}
    V_{0} &:= \frac{1}{2} A_{0} x^{2} + B_{0} x y + \frac{1}{2} C_{0} y^{2} + D_{0} x z + E_{0} y z + \frac{1}{2} F_{0} z^{2} + K_{0}.
    \end{align*}    
 \end{enumerate}  
\end{proposition}

The proof is sketched in Appendix~\ref{se:proofs}.

\subsection{Closed-form Solution for Constant Liquidity}\label{se:cstLiq}

In this section we assume that the liquidity parameters $\beta,\lambda,\eps\in (0,\infty)$ are constant (and in particular $\lambda_{0-} = \lambda$). The coefficients $\theta_{t},\sigma_{t}$ of the in-flow can still be time-dependent. In that setting, we can obtain the optimal strategy in closed form. This solution is based on a different approach, the stochastic maximum principle, which avoids the Riccati system~\eqref{eq:riccati}. On an abstract level, the dynamic programming approach implies the conclusion of the stochastic maximum principle in our context. However, the stochastic maximum principle focuses on the optimal strategy along the optimal state trajectory, and sometimes leads to a closed-form solution even if dynamic programming does not. 

\begin{theorem} \label{th:constParam}
   Let $\beta,\lambda,\eps\in (0,\infty)$ be constant and define also
\begin{equation*}%
      \teps := \frac{\eps \beta}{2 \lambda}, \qquad \kappa := \beta \sqrt{1 + \teps^{-1}} > \beta.
\end{equation*}
Then the functions $f_{t}, g_{t}, h_{t}$ of Theorem~\ref{th:main} are given by
\begin{align}
    f_{t} &= \tilde{f}_{t} (e^{\kappa (T-t)} - 1) / d_{t},   \\
    g_{t} &= \tilde{g}_{t} (e^{\kappa (T-t)} - 1) / d_{t},   \\
    h_{t} &= f_{t} (1-e^{- \int_{t}^{T} \theta_{u} \, du})
\end{align}
for all $t\in[0,T]$, where
\begin{align*}
    \tilde{f}_{t} :&= -[\beta^{-1} - (\kappa +\beta)^{-1}] e^{-\kappa  (T-t)} - [\beta^{-1} + (\kappa - \beta)^{-1}], \\
    \tilde{g}_{t} :&= \lambda^{-1} \tilde{f}_{t} - \lambda^{-1}(1 + e^{-\kappa (T-t)}) (T-t) + 2 \lambda^{-1} \kappa^{-1} (1 - e^{-\kappa (T-t)}), \\
    d_{t} :&= e^{\kappa (T-t)} \left\{ (\kappa - \beta)^{-1}[(\kappa-\beta)^{-1} + \beta^{-1} - \kappa^{-1}] + \beta^{-1} \kappa^{-1} \right\} \nonumber \\
    &\quad+ e^{-\kappa (T-t)} \left\{ (\kappa +\beta)^{-1}[-(\kappa+\beta)^{-1} + \beta^{-1} + \kappa^{-1}] + \beta^{-1} \kappa^{-1} \right\} \nonumber \\
    &\quad+ (T-t) e^{\kappa (T-t)} (\kappa - \beta)^{-1}
    + (T-t) e^{-\kappa (T-t)} (\kappa +\beta)^{-1}
    + 4 \teps \beta^{-1} \kappa^{-1}.
\end{align*}
Moreover, $f_{0-}=f_{0}$, $g_{0-}=g_{0}$, $h_{0-}=h_{0}$.
\end{theorem}

The proof is reported in Appendix~\ref{se:maxPrinciple}.
By specifying the functions $f_{t}, g_{t}, h_{t}$, Theorem~\ref{th:constParam} gives the optimal strategy $q$ in feedback form (as well as the optimal block trade). This may seem surprising, as the maximum principle generally only yields the control in open-loop form, along the optimal trajectory. Indeed, to derive Theorem~\ref{th:constParam}, the maximum principle is carried out from an arbitrary initial state $(t,x,y,z)$ and then the feedback function is recovered by identifying the initial value of the open-loop control.

We note that this approach does not immediately yield a closed-form expression for the cost $\cC(J_{0}, q, J_{T})$ or the functions $A_{t}, B_{t},\dots$ in Theorem~\ref{th:main}. That would require an integration that we cannot carry out in closed form.

\section{Numerical Results}\label{se:numRes}

In this section, we study key features of our model and its optimal strategy by simulation. In addition, we introduce new metrics of interest to practitioners. Section~\ref{se:numSetup} describes the methodology, model parameters, and trading metrics.
The most essential trading question is how unwind strategy and cost metrics depend on the characteristics of the in-flow; namely, autocorrelation and volatility. Sections~\ref{se:numAutocorr} and~\ref{se:numVola} answer this question and show that different parameters give rise to different trading regimes ranging from optimal execution to market making. Our model covers the full range of scenarios. A crucial insight from the study is that practitioners must carefully model their in-flow to obtain useful results. Therefore, the mantra ``know your client''
is a recurring theme. Section~\ref{se:numMonoton} proposes a metric capturing how \emph{monotone} a given in-flow path is. Monotonicity, in turn, correlates strongly with internalization and trading cost under the optimal trading strategy.

Several further experiments are reported in Appendix~\ref{se:appendixNumerical}. Section~\ref{se:numDeterministic} links our model to classical (deterministic) models of optimal execution and illustrates that the Obizhaeva--Wang model is recovered for small spread cost parameter $\eps\to0$.  Section~\ref{se:numDeterministic} also studies the sensitivity with respect to~$\eps$, both in the deterministic and stochastic case. Section~\ref{se:itoLimit} analyzes the sensitivity with respect to the martingale driver of the in-flow as the martingale changes  from a finite-variation process with a small number of shocks to a continuous-time Brownian motion. Section~\ref{se:initalCondition} studies the effect of a non-zero initial impact state on the trading strategy. Lastly, Section~\ref{se:empiricalAutocorr} presents an empirical study on autocorrelation of orders on the public trading tape, confirming momentum for all S\&P 500 stocks.

\subsection{Numerical Setup}\label{se:numSetup}

We focus on constant parameters throughout our experiments. While time-varying parameters are important in a practical implementation, constant parameters allow for more interpretable results and a straightforward way of varying parameter values in our experiments. %

The theoretical results of the preceding section enable an efficient numerical implementation: the feedback form of the optimal strategy (i.e., the coefficients $f_{t},g_{t},h_{t}$ and the block trade~$J_0$) can be precomputed independently of the random inputs.  As parameters are constant, we use the closed-form formulas outlined in Section~\ref{se:cstLiq}. (The ODEs for the time-varying case are also straightforward to implement and fully precomputable). Given the optimal strategy, the simulation 
of the state processes $(X,Y,Z)$ is straightforward by discretization of their SDEs. As a result, large-scale numerical experiments using many sample paths can be run without heavy-duty hardware. The code is provided on GitHub.\footnote{\url{https://github.com/KevinThomasWebster/UnwindingStochasticOrderFlow/blob/main/simulation.ipynb}}

One difference with the main part of Section~\ref{se:theory} is that we focus on an in-flow process of finite variation. For practitioners, the most important metric in our context is the internalization rate, defined in~\eqref{eq:defInternalizationRate} below. Its formula includes a normalization by the total variation of the in-flow process, or in other words, the cumulative size of all orders received. If the in-flow is modeled using a Brownian motion driver as in Theorem~\ref{th:main}, the total variation is infinite and the internalization rate cannot be defined. Instead, our experiments use a discrete driver where the inventory is shocked by a block trade approximately every 20 minutes. We recall from Section~\ref{se:jumps} that our model is solvable for general martingale drivers. The shocks are taken to be i.i.d.\ Gaussian (any centered distribution could be used). While the main purpose is to have a well-defined internalization rate, this also adds realistic block orders to our simulation.  Section~\ref{se:itoLimit} discusses the robustness of our results with respect to this discretization.

\subsubsection{Inputs and Methodology}

We use standard units to normalize our variables across securities. The time unit is days; we set  $T =1$ so that the strategy trades daily over the interval $[0,1]$.\footnote{In our plots, we nevertheless label time with the standard 9:30--16:00 hours.} Denote by $[T]$ the time unit.
All trading quantities are expressed in percent of Average Daily Volume (ADV\%). For instance, a jump in the in-flow of 2\% means that the desk received an order of size 2\% ADV. In-flow volatility of $10\%$ indicates that the trading team expects 10\% of the ADV to go through their desk over the course of the day. We denote by $[V]$ the volume unit.

Next, we describe the trading inputs and their default values in our simulations.

\paragraph{Impact and spread parameters.} 
    \begin{enumerate}
    \item Impact decay $\beta = 8$ has unit $[T]^{-1}$ and describes the speed at which impact reverts. The corresponding half-life is approximately 40 minutes. 
    \item Kyle's lambda $\lambda = 0.2$ has unit $[V]^{-1}$ and describes an individual fill's impact.\footnote{For empirical studies, Kyle's lambda is often normalized for the stock's volatility. For our numerical study, we picked a stock with moderate volatility.}
    \item The spread cost parameter $\eps = 0.01$ has unit $[V]^{-2}[T]^{-1}$ and describes in reduced form the various instantaneous trading costs, such as the bid-ask spread and adverse selection.
    \end{enumerate}

These parameters were picked to match the overall statistics (spread and impact cost, volume on close) from Nasdaq's 2022 Intern's Guide to Trading \cite{Mackintosh2022} and align with the empirical estimates found in \cite{Chen2019} and  \cite{JMK2022b}. 

\paragraph{In-flow parameters.}

The in-flow process uses 20 equispaced i.i.d.\ Gaussian shocks (rather than a Brownian motion).\footnote{This process is not realistic for order flow; it is chosen to make the figures easy to interpret and to have a straightforward simulation code. Random arrival times and self-exciting volume would be more realistic. However, the take-aways for more complicated models would be similar as we know from our theoretical results that the optimal feedback strategy does not depend on the driving martingale.} The volatility parameter $\sigma$ controls the variance of those shocks. We express $\sigma$ in terms of daily quadratic variation to make it comparable to the Brownian motion case or a different number of intraday shocks.

    \begin{enumerate}
    \item The in-flow volatility $\sigma = 0.1$ has unit $[V] [T]^{-1/2}$ and describes in-flow shocks.
    
    \item The in-flow's autocorrelation (Ornstein--Uhlenbeck) parameter $\theta$ has unit $[V] [T]^{-1}$. We consider three cases: the default martingale case $\theta=0$, the momentum case $\theta = -1$, and the reversal case $\theta = 1$. Therefore, in the momentum case, if the desk stopped trading the out-flow, their inventory would in expectation double over $\log(2) \approx 0.69$ days. Conversely, in the reversal case, if the desk stopped trading the out-flow, their inventory would naturally internalize by half over $\log(2) \approx 0.69$ days.
    \end{enumerate}

\paragraph{Initial conditions.}
    \begin{enumerate}
        \item The initial inventory $x = -z = -0.1$ has unit $[V]$.
        \item The initial impact state $y = 0$ is unitless (but can be normalized in terms of stock volatility in an empirical study).
    \end{enumerate}
    Non-zero initial conditions matter in practice. While the market-making literature often assumes zero initial position, the desk cannot assume that its portfolio is always trading in a steady state. The optimal execution literature obviously uses non-zero initial inventory but often assumes zero initial impact. In practice, impact from the previous day's trading cannot be neglected and is similar to an alpha signal: the strategy must take the reversal of prices into account (see also Section~\ref{se:initalCondition}). Our model and the provided code allow for non-zero initial conditions.

\subsubsection{Trading Metrics}\label{se:tradingMetrics}

Following industry practice, we normalize trading metrics by the total variation of the in-flow (i.e., the cumulative order notional). This corresponds to the \emph{client's perspective} as the metric can be understood as ``per order notional.''

\paragraph{In-flow, out-flow, inventory (ADV\%).}

We plot the intraday time series of flows and inventories. As explained in the beginning of Section~\ref{se:numSetup}, we choose an in-flow with finite variation for our simulations.

\paragraph{Internalization (\%).}

The internalization rate is defined~\cite[p.\,29]{BIS.22} as 
\begin{equation}\label{eq:defInternalizationRate}
  \mbox{internalization} = 1-\frac{\mbox{total variation of out-flow}}{\mbox{total variation of in-flow}};
\end{equation}
roughly speaking, it is the fraction of in-flow orders that were netted. From the client's point of view, high internalization indicates that many orders were netted \emph{instead} of trading in the market and thus that the desk was able to reduce trading costs. The internalization rate is upper bounded by one, but this bound is achieved only if there is no out-flow (or if the in-flow has infinite total variation). Zero internalization is achieved in particular if out-flows match in-flows, meaning that the desk routes all orders to the market. For a given in-flow, the highest internalization rate is obtained by a monotone out-flow (the desk only sells or only buys in the market), as that minimizes the total variation given the liquidation constraint. In principle, internalization can be negative, but it usually takes a value in $[0,1]$ because the desk chooses an out-flow smoother than their in-flow. 

Practitioners consider internalization rate as a proxy for saved trading costs (but see our discussion of Figure~\ref{fig:internProxy} below for a caveat). Compared to cost, internalization rate has the advantage of being model-free, which may explain why it is the most important metric for clients and regulators (cf.\ \cite[Footnote~23]{BIS.22}).

\paragraph{Internalization regret (\%).}

We propose a novel metric called internalization regret which focuses on the \emph{desk's perspective} and measures how many trades were ex-post unnecessary.
We define internalization regret as\footnote{A similar definition would be $(\text{total variation of out-flow)/|\mbox{terminal in-flow}|}$. That quantity is numerically ill-behaved as the terminal in-flow may be close to or equal to zero.}
\begin{equation}\label{eq:defInternalizationRegret}
  \mbox{internalization regret} = 1- \frac{|\mbox{terminal in-flow}|}{\mbox{total variation of out-flow}};
\end{equation}
roughly speaking, it is one minus the ratio of minimally necessary over executed out-flow trades.
Note that we use the absolute terminal in-flow $|Z_{T}|$ rather than the total variation of the in-flow process $Z$. Because the total variation of any out-flow satisfying the liquidation constraint is at least $|Z_{T}|$, internalization regret always takes values in $[0,1]$. (We use the convention $0/0:=1$.) It has value zero if and only if the out-flow is monotone (the desk only sells or only buys in the market), indicating that the desk captured all possible internalization. In particular, zero regret is always achieved by the naive strategy that warehouses all in-flow during the day and then externalizes the inventory on the close. On the other hand, a high regret indicates that the desk traded in the market where netting would (ex-post) have been possible. 

In contrast to the internalization rate, the internalization regret is well-defined even if the in-flow has infinite variation (see also Section~\ref{se:itoLimit}).

\paragraph{Impact, spread, total costs (bps).}

The absolute impact cost was defined in~\eqref{eq:impactCostC0} in continuous time; it is discretized as
\begin{equation*}
y J_0 + \frac{\lambda}{2}J^2_0 + \sum_t Y_t q_t \Delta t +  Y_{T-\Delta t} J_T + \frac{\lambda}{2}J^2_T.
\end{equation*}
As we measure trading costs from the clients' perspective in this section, we report the absolute cost divided by the total variation of the in-flow. Similarly, the  absolute spread (or instantaneous) cost defined in~\eqref{eq:spreadCostC0} is discretized as $\frac{1}{2}\sum_t \eps q^2_t \Delta t$ and divided by the total variation of the in-flow.

\paragraph{Closing trade (ADV\%, \%).}

The closing trade is the block trade placed in the closing auction. We express it in two units: either in ``absolute'' terms (ADV\%) or as a proportion of the (total variation of the) out-flow (\%). The former is more relevant to the client, as it directly compares to the in-flow. On the other hand, the latter captures how much the desk warehouses until the close compared to its overall trading activity.

\subsection{Sensitivity to In-Flow Auto-Correlation}\label{se:numAutocorr}

We first study how the optimal out-flow reacts to the autocorrelation parameter ($\theta$). The figures show martingale, reversal, and momentum in-flow. 
Figure~\ref{fig:ACSamplePath} shows a sample path of the trading strategy's intraday evolution. As expected, the in-flow exhibits larger variations with momentum dynamics and lower variations with reversal dynamics. Therefore, the out-flow and impact state is most aggressive for momentum and least aggressive for reversal in-flow.
Figure~\ref{fig:ACAvgPath} shows the average path of the trading strategy. In expectation, the out-flow and impact time series are identical to the deterministic optimal execution case when accounting for the expected future in-flow (cf.\ Section~\ref{se:numDeterministic}).

\begin{figure}[htbp]
    \centering
    \includegraphics[width=1.0\textwidth]{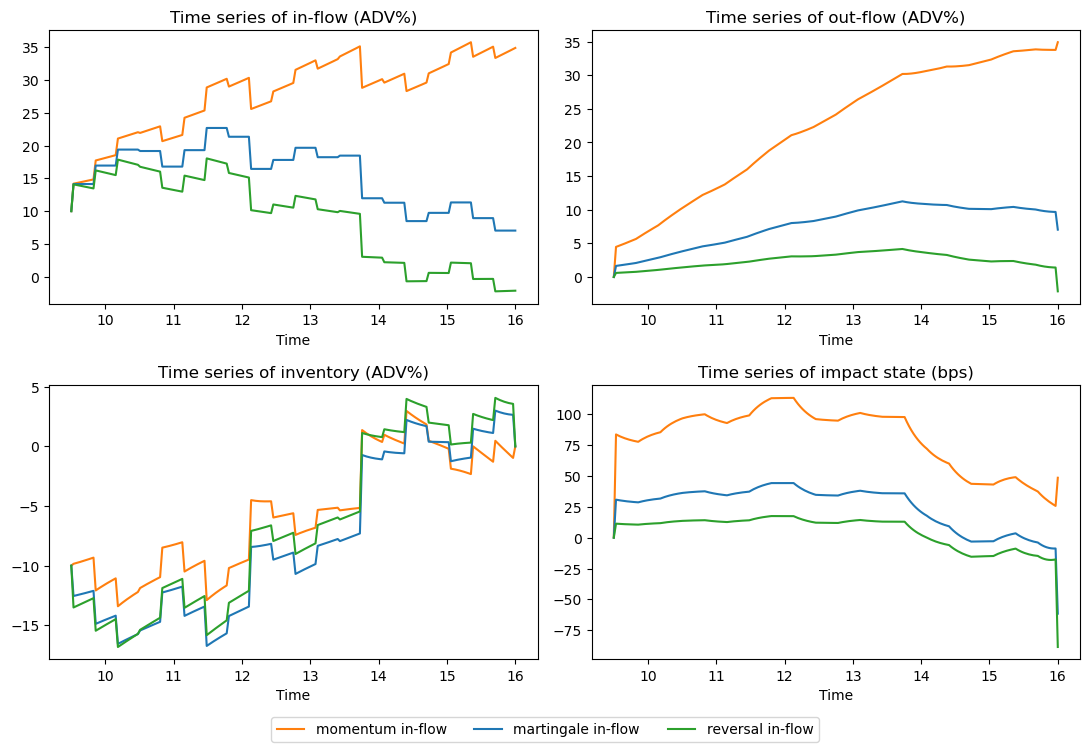}
    \caption{Sensitivity to flow autocorrelation ($\theta$) for an initial inventory ($z$) and daily flow volatility ($\sigma$) of 10\% ADV: sample path.}
    \label{fig:ACSamplePath}
\end{figure}

\begin{figure}[htbp]
    \centering
    \includegraphics[width=1.0\textwidth]{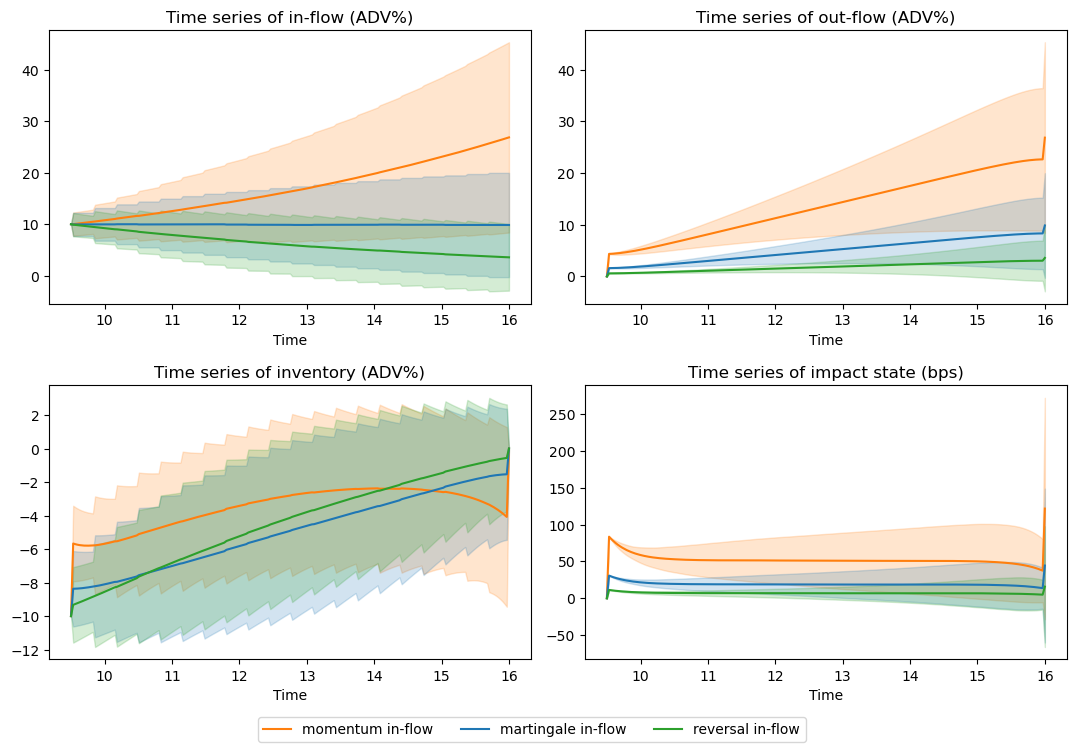}
    \caption{Sensitivity to flow autocorrelation ($\theta$) for an initial inventory ($z$) and daily flow volatility ($\sigma$) of 10\% ADV: average path. The error bars illustrate one standard deviation around the average path.}
    \label{fig:ACAvgPath}
\end{figure}

\begin{table}[tbhp]
    \centering
    \tabcolsep=0.15cm 
    \begin{tabular}{ ||c c | c c c c c c||}
     \hline
     Parameter      & $\theta$ & In-flow & Out-flow & Spread Cost  & Impact cost & Closing trade & Internalization \\ 
      scans &          & (ADV\%)  & (ADV\%)&  (bps)       &  (bps)      & (\% total)       & (\%)\\[0.5ex]
     \hline
     \hline
     momentum   & -1 & 61 & 31 & 4.9 & 42.6 & 17 & 51 \\ 
     martingale & 0  & 46 & 15 & 1.7 & 14.5 & 21 & 68 \\ 
     reversal   & 1  & 52 & 9  & 0.5 & 4.8  & 27 & 84 \\ [0.5ex]
     \hline  
    \end{tabular}
    \caption{Sensitivity to flow autocorrelation ($\theta$) for an initial inventory ($z$) and daily flow volatility ($\sigma$) of 10\% ADV: average metrics.}
    \label{tbl:AC}
\end{table}

Table \ref{tbl:AC} reports the average daily statistics. We observe that out-flows are significantly more sensitive to $\theta$ than in-flows. Consequently, in-flow autocorrelation strongly affects all core trading metrics:
\begin{enumerate}
\item Higher autocorrelation leads to higher costs. The impact-to-spread ratio remains stable. Therefore, higher out-flows drive the cost increase, which is proportionally shared by spread and impact.
\item Lower autocorrelation leads to higher internalization. In particular, the trading strategy warehouses more inventory with lower autocorrelation, which leads to a larger closing trade as a proportion of out-flow.
\end{enumerate}

\FloatBarrier

Figure~\ref{fig:ACDistr} shows the distribution of daily trading metrics across samples of the in-flow path. As expected, distributions are wide, highlighting the need for stochastic control theory over deterministic optimization techniques: the deterministic optimal execution problem only describes the strategy's \emph{average} behavior.

\begin{figure}[htbp]
    \centering
    \includegraphics[width=1.0\textwidth]{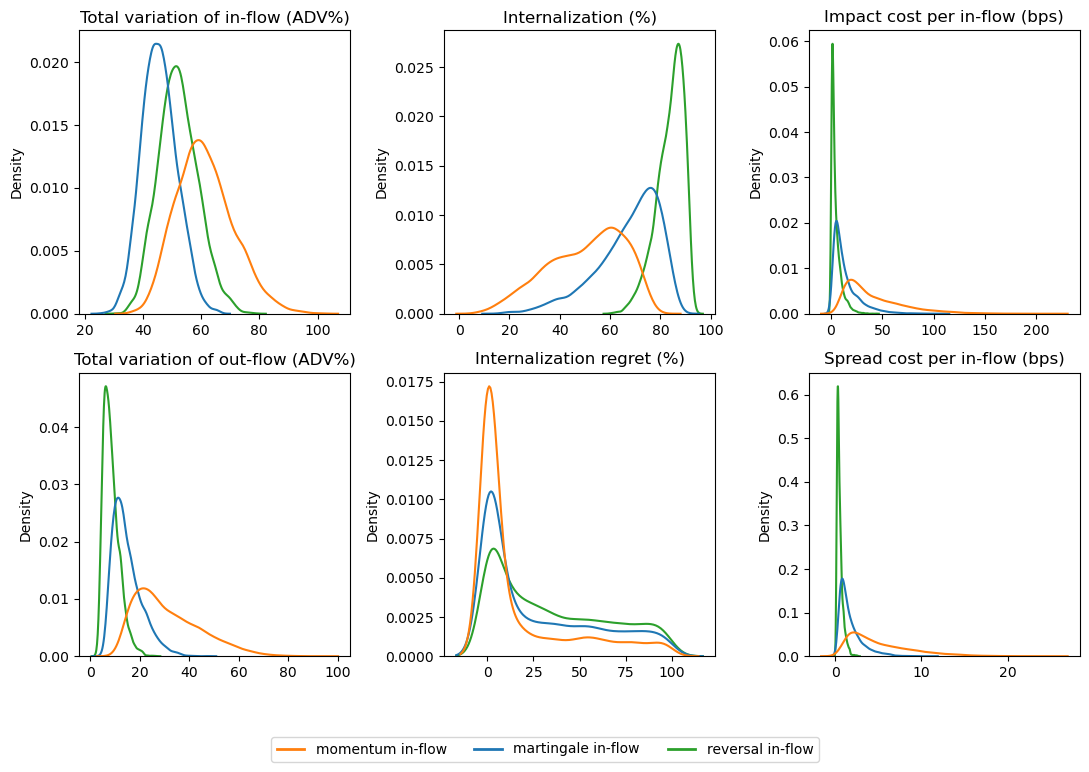}
    \caption{Sensitivity to flow autocorrelation ($\theta$) for an initial inventory ($z$) and daily flow volatility ($\sigma$) of 10\% ADV: metrics' distribution.}
    \label{fig:ACDistr}
\end{figure}

As the first panel of Figure~\ref{fig:ACDistr} illustrates, martingale in-flow has the smallest total variation. Indeed, any drift must increase total variation, but momentum more so than reversal. This is relevant as the in-flow's total variation is used to normalize trading metrics. Nevertheless, internalization is highest in the reversal case and lowest in the momentum case. This is explained by the total variation of the out-flow: reversal naturally leads to netting while momentum requires more trading. We observe that the distribution of the internalization rate is significantly wider for momentum than for reversal. The impact cost shows largely the same characteristics as the internalization. %

The distribution of the internalization regret has an atom at zero, decays away from zero, and becomes approximately uniform for large regret values. The mass at zero reflects monotone trading paths (the desk only sells or only buys). The concentration is strongest in the momentum case, where the probability of zero regret is about 23\% and the probability of regret below 1\% is about 46\%, compared with (11\%, 26\%) in the martingale case and (4\%, 13\%) for reversal. (Density plotted outside [0,100] is an artifact of kernel density estimation.) Generally, increasing autocorrelation lowers internalization and increases costs.

\begin{figure}[htbp]
    \centering
    \includegraphics[width=0.5\textwidth]{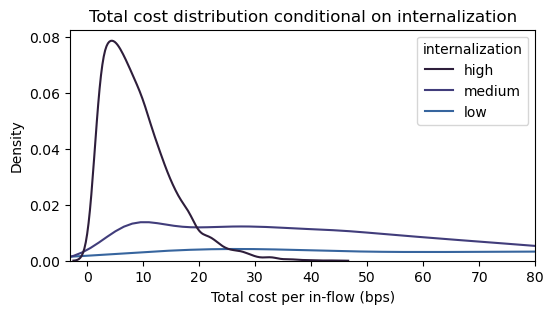}
    \caption{Total cost distribution conditional on internalization. Samples are taken across a wide range of in-flow parameters.}
    \label{fig:internProxy}
\end{figure}

Figure~\ref{fig:internProxy} highlights an important rule of thumb: high internalization is a strong proxy for low costs, regardless of the shape of the in-flow path. On the other hand, when internalization is low, costs can take a wide range of values depending on the in-flow path, meaning that internalization alone is not a suitable proxy for the cost.

\FloatBarrier

\subsection{Sensitivity to In-Flow Volatility}\label{se:numVola}

Figure~\ref{fig:volScan} shows the sensitivity of key metrics to in-flow volatility ($\sigma$). We observe that the relation between initial inventory and in-flow volatility is key.
In the regime where the initial inventory is significantly larger than the daily volatility of the in-flow, the trading strategy behaves like an optimal execution strategy.
Conversely, if the in-flow volatility dominates the initial inventory, the trading strategy behaves like a steady-state market-making strategy.
An important transition occurs between those regimes:

\begin{enumerate}
\item Before the critical value, increasing the in-flow volatility increases internalization and drives down (per in-flow) costs.
\item After the critical value, internalization reaches a plateau. Therefore, in-flow volatility drives up total trading, and costs increase.
\end{enumerate}

\begin{figure}[htbp]
    \centering
    \includegraphics[width=\textwidth]{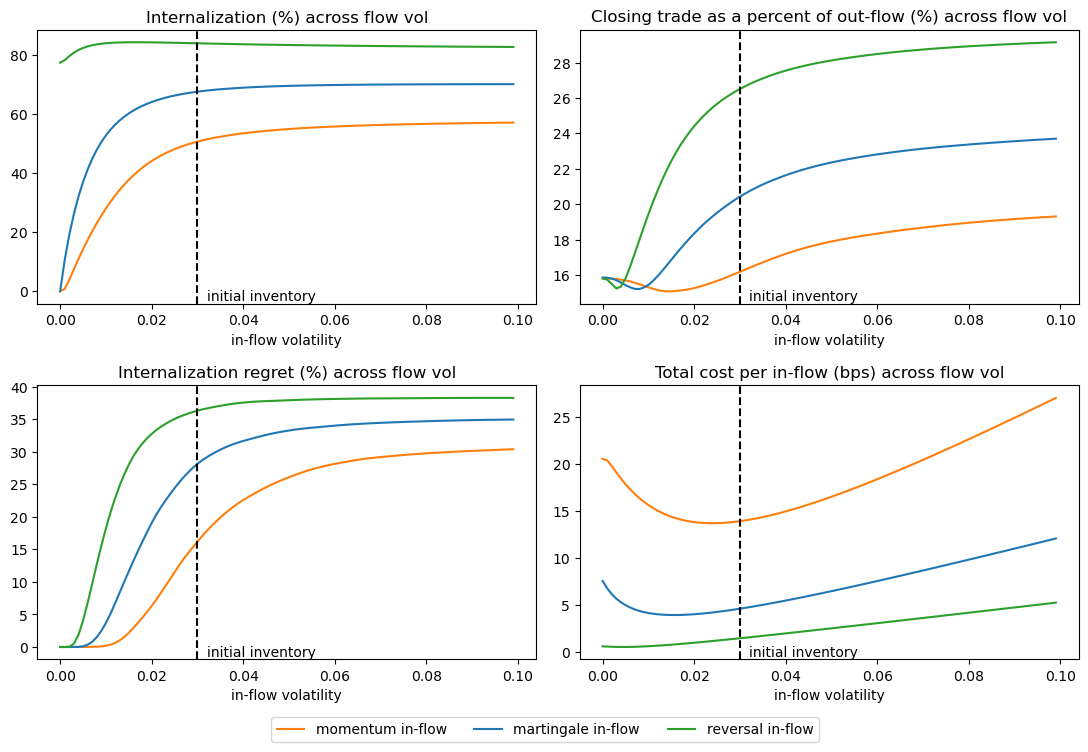}
    \caption{Average metrics' sensitivity to in-volatility ($\sigma$) for an initial inventory ($z$) of 3\% ADV and different autocorrelation values ($\theta$).}
    \label{fig:volScan}
\end{figure}

As an important consequence, a desk with low in-flow volatility may benefit from attracting additional in-flows, e.g., by onboarding new clients or posting more attractive prices, in order to increase internalization and reduce costs, whereas a desk with high volatility may benefit from lowering it. We observe that the transition happens when volatility and initial position are comparable in size, with the exact location depending on the in-flow autocorrelation. Table~\ref{tbl:volScan} quantifies the optimal in-flow vol to attract for a given in-flow autocorrelation. The in-flow vol scales linearly with the initial inventory, and, therefore, we quote the ratio.\footnote{The optimal in-flow volatility is insensitive to the martingale driver: there is no material difference between the diffusion and jump cases. The fact that the location of the cost minimum depends only on the ratio $\sigma/z$ can be seen analytically by a scaling argument. If the initial impact state~$y$ is nonzero, $y$ also needs to be scaled, so that the location of the minimum depends on the two ratios $\sigma:z:y$.}

\begin{table}[htbp]
    \centering
    \tabcolsep=0.15cm 
    \begin{tabular}{ ||c | c c c c c||}
     \hline
     $\theta$      & -1 & -0.5 & 0 & 0.5 & 1 \\ [0.5ex]
     \hline
     \hline
     optimal $\sigma/z$ (\%)   & 99 & 79 & 55 & 35 & 17 \\ 
     Internalization (\%)      & 55 & 65 & 72 & 77 & 82 \\  [0.5ex]
     \hline  
    \end{tabular}
    \caption{Optimal in-flow volatility, as a percentage of inventory, across various $\theta$.}
    \label{tbl:volScan}
\end{table}

\FloatBarrier

\subsection{Internalization and  In-Flow Monotonicity}\label{se:numMonoton}

Opportunities for internalization depend on the in-flow path, which may trend in a single direction or net out naturally. 
To capture this, we propose a measure of monotonicity of in-flow paths, namely the terminal absolute value of in-flow normalized by the total variation of in-flow ($|Z_{T}|/\|Z\|_{TV}$). A value of 100\% corresponds to an in-flow that is monotone; i.e., consist of only buy or only sell orders.  
Figure~\ref{fig:monotonicity} shows that this metric is closely related to internalization under the optimal out-flow strategy. To produce Figure~\ref{fig:monotonicity}\,(a), we sample 100'000 paths from our model and plot the internalization rate and regret of the out-flow versus the normalized terminal in-flow. The more ``monotone'' the in-flow, the lower the internalization rate and internalization regret. We emphasize that some parts of the x-axis correspond to more sample paths than others. Indeed, the distribution of $|Z_{T}|/\|Z\|_{TV}$ depends on the in-flow volatility ($\sigma$), as shown in Figure~\ref{fig:monotonicity}\,(b). 

\begin{figure}[htb]
	\center
	\begin{subfigure}[b]{\textwidth}
		\includegraphics[width=\textwidth]{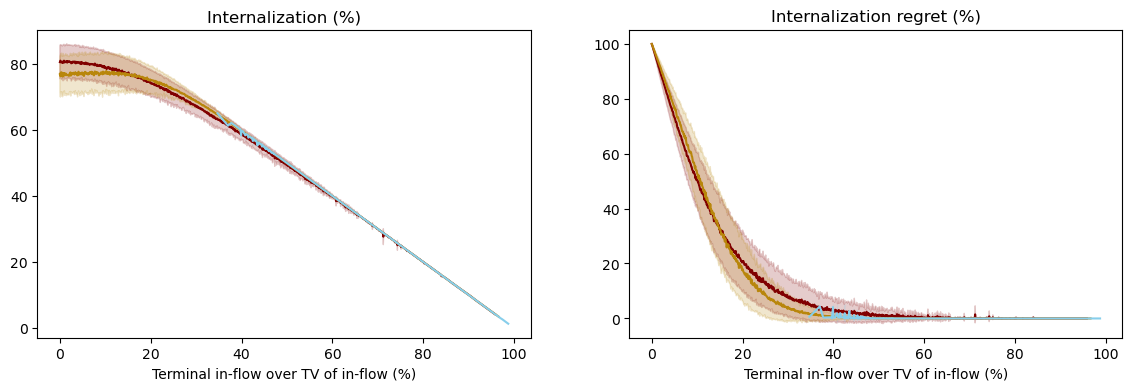}
		\caption{Normalized terminal in-flow ($|Z_{T}|/\|Z\|_{TV}$) is a strong predictor of internalization.}
	\end{subfigure}
	\\
	\begin{subfigure}[b]{0.5\textwidth}
		\includegraphics[width=\textwidth]{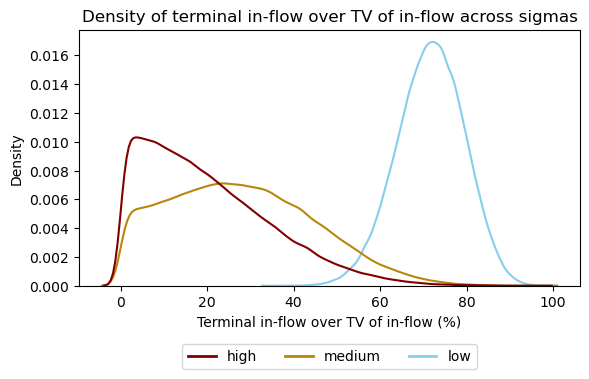}
		\caption{In-flow volatility affects the distribution of the normalized terminal in-flow  ($|Z_{T}|/\|Z\|_{TV}$).}
	\end{subfigure}
\caption{(a) Monotonicity test: trading metrics as a function of the terminal in-flow state, as a fraction of in-flow TV. The error bars illustrate one standard deviation around the average. (b)~The distribution of the proposed metric depends on in-flow volatility.}
	\label{fig:monotonicity}
\end{figure}

In practice, the shape of the in-flow path depends largely on the desk's clients. The strong relationship in Figure~\ref{fig:monotonicity} demonstrates that the nature of the clients---and not only the desk's unwind strategy---is a key determinant for the internalization rate. Therefore, the trading team gains a lot from understanding their clients. Furthermore, our proposed monotonicity metric is a strong, model-free client metric forecasting achievable internalization rates.

\FloatBarrier

\subsection{Misspecification Costs}\label{se:misspecification}

This section quantifies the cost of \emph{not} knowing your client, or more precisely, the cost of trading with the wrong in-flow parameter~$\theta$. Thanks to the analytic structure provided by Theorem~\ref{th:main}, this misspecification cost can be computed in two steps: first, compute the strategy under a given (misspecified) estimate $\hat\theta$, then, simulate the state processes under a ground-truth value $\theta^*$. The strategy becomes suboptimal unless $\hat\theta = \theta^*$.

\begin{figure}[htb]
	\center
	\begin{subfigure}[b]{\textwidth}
		\includegraphics[width=\textwidth]{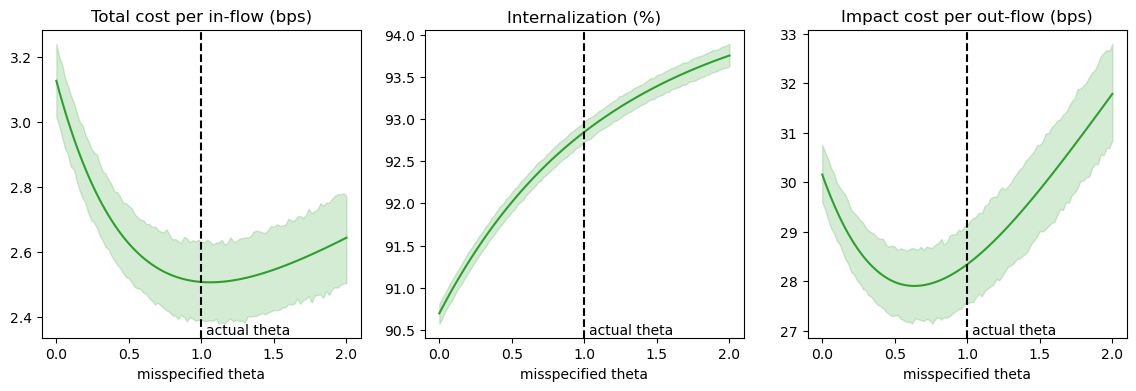}
		\caption{$\theta^* = 1$. The actual in-flow is mean-reverting.}
	\end{subfigure}
	\\
	\begin{subfigure}[b]{\textwidth}
		\includegraphics[width=\textwidth]{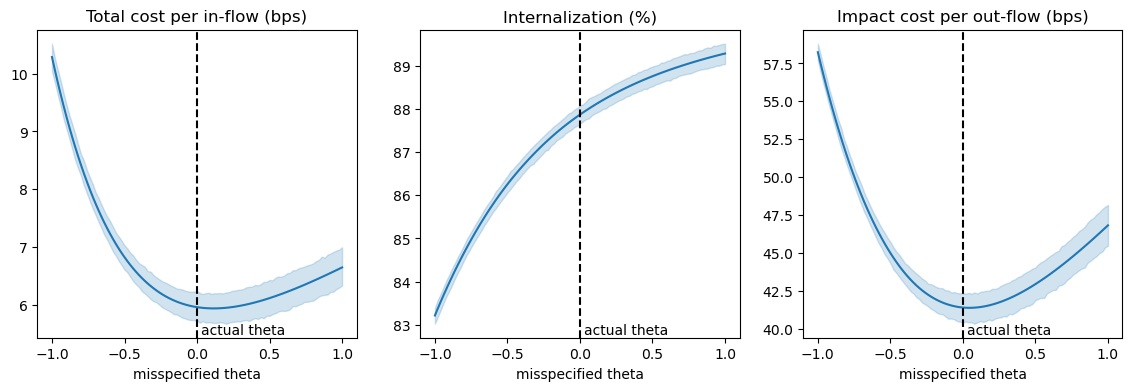}
		\caption{$\theta^* = 0$. The actual in-flow is a martingale.}
	\end{subfigure}
   \\
	\begin{subfigure}[b]{\textwidth}
		\includegraphics[width=\textwidth]{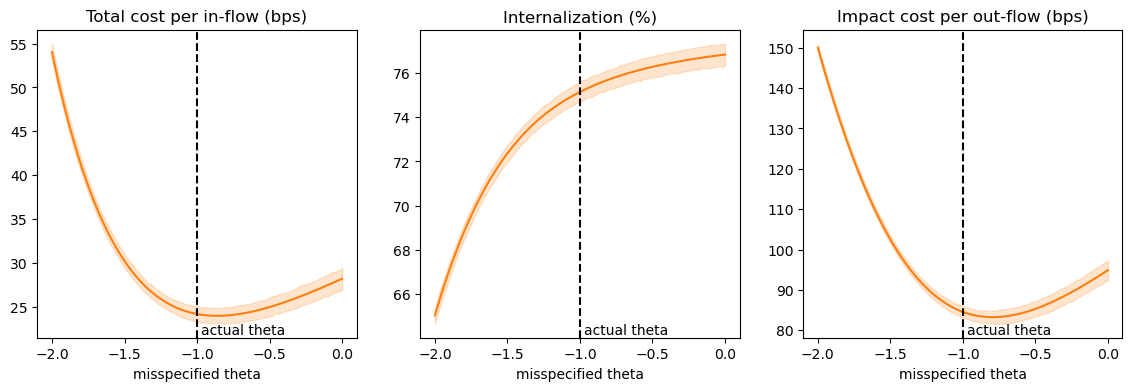}
		\caption{$\theta^* =-1$. The actual in-flow has momentum.}
	\end{subfigure}
\caption{Simulated cost per in-flow, internalization, and impact per out-flow across various $\theta^*, \hat\theta$.}
	\label{fig:misspecificationExamples}
\end{figure}

Figure~\ref{fig:misspecificationExamples} studies three misspecification examples with the following key conclusions:
\begin{enumerate}
\item Misspecification costs are highest for momentum and lowest for mean-reverting flow. Getting $\theta$ wrong by $100\%$ increases costs by less than a third for mean-reverting flow. By contrast, misspecifying $\theta$ by $100\%$ for momentum flow more than doubles costs.
\item Misspecification costs are asymmetric: for a given $\theta^*$, one would rather over- than under-estimate $\theta$. This is in line with the misspecification behavior for impact cost in \cite{Hey2023}: generally, if a model parameter is uncertain, one prefers to err on the side which leads to more conservative trading.
\item Both internalization and impact costs increase drastically when $\hat\theta$ moves away from $\theta^*$. Therefore, both channels drive the misspecification costs.
\end{enumerate}

One concrete application of this misspecification study relates to estimates of in-flow autocorrelation. All else being equal, given a statistical confidence interval over~$\theta$ (or any other measure of in-flow autocorrelation), it is best to deploy the most conservative trading strategy, that is, the strategy that underestimates in-flow momentum.

There is no misspecification cost for~$\sigma$ since the feedback form of the optimal strategy does not depend on~$\sigma$ (cf.\ Proposition~\ref{pr:properties}\,(iv)). Of course, the expected cost is nevertheless affected, hence the quote given to the client may bear a misspecification error.

We do not conduct a misspecification study for the liquidity parameters; the existing results~\cite{Neuman2023, Hey2023, Neuman2023bis} apply to our setting.

\FloatBarrier

\newpage
\appendix
\section{Further Numerical Results}\label{se:appendixNumerical}

\subsection{Recovering Deterministic Results; Sensitivity to Spread Cost Parameter}\label{se:numDeterministic}

This section serves two purposes. First, we show how our model links to the traditional optimal execution literature where in-flow is deterministic. Second, we analyze the sensitivity of our model to the spread cost parameter ($\eps$). The setup is as described in Section~\ref{se:numSetup}.

The first experiment is best understood as a minimal extension of the traditional optimal execution literature. We consider the deterministic case with no in-flow apart from the initial order of $z=10\%$ ADV (i.e., $\sigma=\theta=0$) . Figure~\ref{fig:optimExecTest} reproduces the optimal trading strategy of Obizhaeva and Wang in the limit as the spread cost parameter $\eps \rightarrow 0$. As discussed in \cite{webster.23}, the optimal impact state tracks a certain target value in the Obizhaeva--Wang model, with jumps at the start and end of the day to hit that target. In our model with instantaneous cost during LOB trading, the intraday impact state is smoothened, but the change in target is substantial only for high spread cost parameter.

In the out-flow and inventory, we observe that the size of the block trades in the auctions increases with the spread cost: as there is no spread cost during the auctions, trading a larger volume during the auctions is a way evade high spread cost in LOB trading, partly at the expense of higher impact costs. That also explains why the impact state initially overshoots the target value.

\begin{figure}[htbp]
    \centering
    \includegraphics[width=1.0\textwidth]{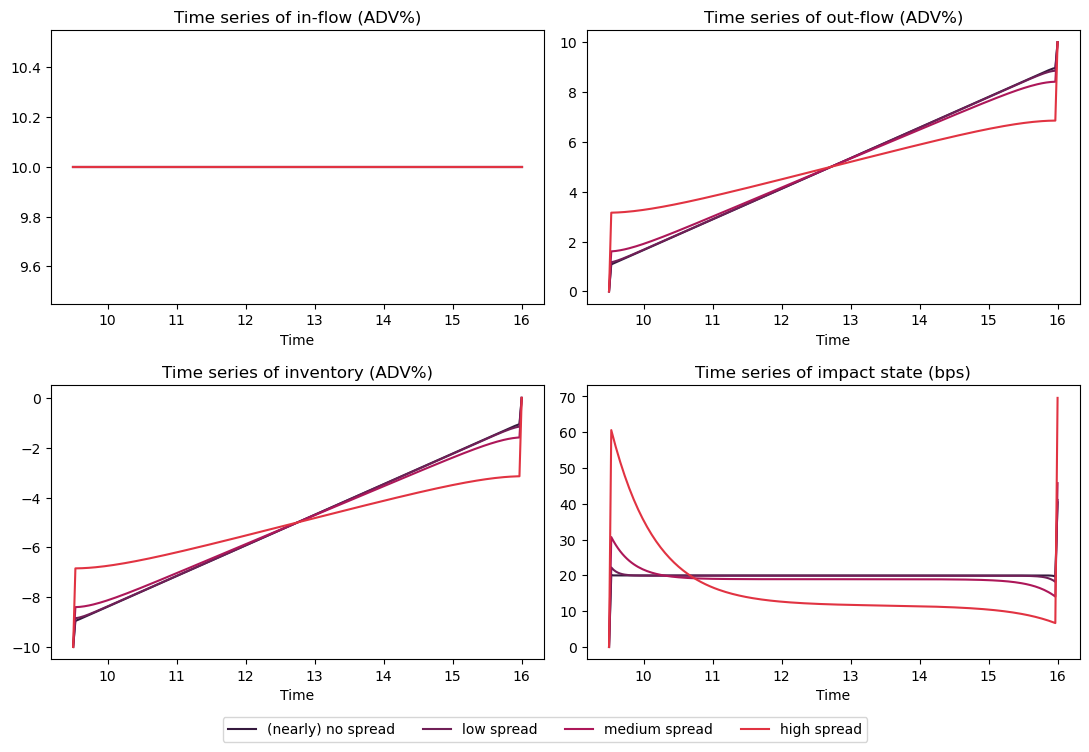}
    \caption{Sensitivity to $\eps$ of the optimal execution strategy for a 10\% ADV order: path.}
    \label{fig:optimExecTest}
\end{figure}

Table \ref{tbl:epsilonSensitivity} quantifies some of these insights and highlights the importance of the impact-to-spread cost ratio. 
When impact costs outweigh spread costs, $\eps$'s effect is second order and does not meaningfully affect the impact state or trading speed. This is the typical regime for a desk dealing with large orders. However, once spread costs become comparable to impact costs, the trading regime changes. Overall costs increase dramatically and substantial part of the inventory is unwound by block trades on open and close.

\begin{table}[htbp]
    \centering
    \begin{tabular}{ ||c c | c c c c c||}
     \hline
     Parameter      & $\eps$ & In-flow  & Spread Cost  & Impact cost & Closing trade & Impact to \\ 
      scans &            & (ADV\%)  &  (bps)       &  (bps)      & (ADV\%)       & spread ratio\\[0.5ex]
     \hline
     \hline
     (nearly) no spread & 1e-4 & 10 & 0.0 & 21.1 & 1.0 & 668 \\ 
     low spread         & 1e-3 & 10 & 0.3 & 21.3 & 1.1 & 71 \\ 
     medium spread      & 1e-2 & 10 & 2.4 & 22.8 & 1.6 & 9 \\ 
     high spread        & 1e-1 & 10 & 7.5 & 36.9 & 3.2 & 5 \\ [0.5ex]
     \hline  
    \end{tabular}
    \caption{Sensitivity to $\eps$ of the optimal execution strategy for a 10\% ADV order: metrics.}
    \label{tbl:epsilonSensitivity}
\end{table}

Next, we study the sensitivity of daily metrics to $\eps$ when the in-flow is random. Connecting to Section~\ref{se:numAutocorr}, 
Figure~\ref{fig:ACSpread} also shows the sensitivity with autocorrelation $(\theta)$. As expected, a higher spread cost parameter $\eps$ leads to higher internalization, lower internalization regret, and larger closing trades as it slows down the continuous trading. This holds for all values of $\theta$, with the ordering among reverting, martingale, and momentum flow being as expected. The main take-away, however, is that internalization and impact cost are relatively insensitive to~$\eps$ as long as values are moderate. Internalization regret is only slightly more sensitive (note that~$\eps$ is being varied over two orders of magnitude in the figure).
 
\begin{figure}[htbp]
    \centering
    \includegraphics[width=1.0\textwidth]{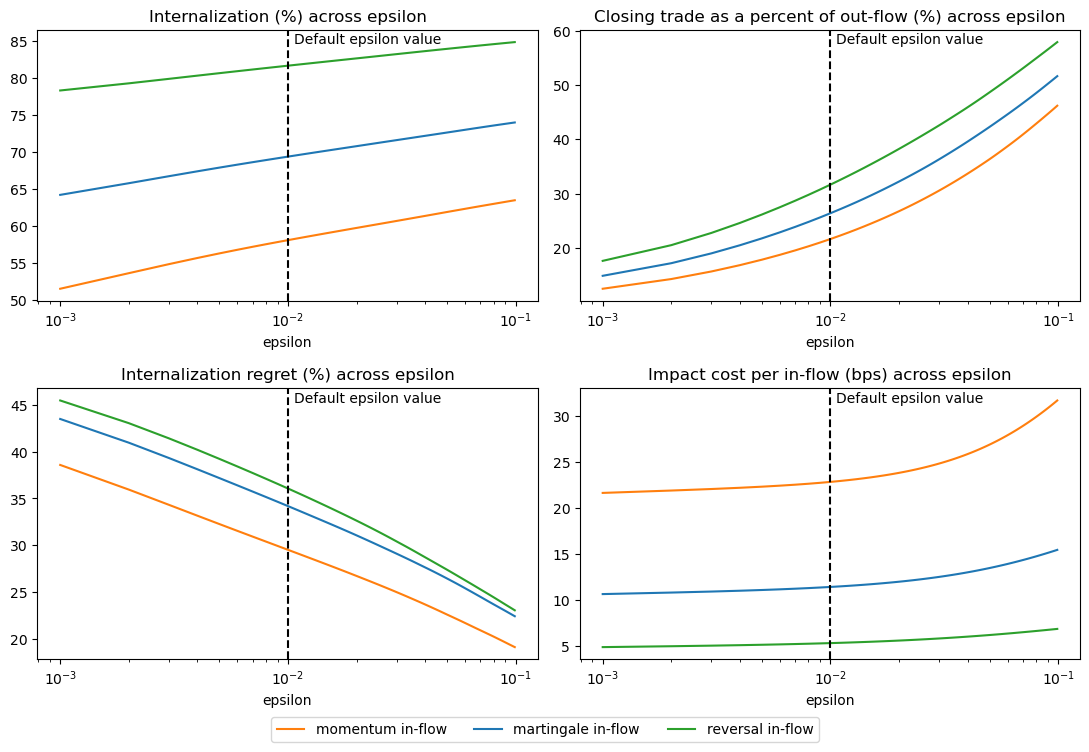}
    \caption{Average metrics' sensitivity to spread costs ($\eps$) for different flow autocorrelations ($\theta$), for an initial inventory ($z$) and daily in-flow volatility ($\sigma$) of 10\% ADV.}
    \label{fig:ACSpread}
\end{figure}

\FloatBarrier

\subsection{It\^o Diffusion Limit for In-Flows}\label{se:itoLimit}

This section revisits the driving martingale of the in-flow process~$Z$. As seen in Section~\ref{se:jumps}, the model is solvable for finite-variation (jump) processes, diffusions, and combinations thereof. The feedback form of the optimal strategy does not depend on the driving martingale, but of course the actual trade path depends on the realization of the state processes. Our preceding numerical experiments use a finite-variation process so that trading metrics can be reported in a standard fashion. Here, we study how the metrics behave in the diffusion limit; i.e., as the number of jumps increases (while their size decreases such as the keep the expected quadratic variation constant).

Figure~\ref{fig:itoLimit} shows the average metrics. The overarching conclusion is that the trading strategies as well as trading metrics remain similar in the diffusive limit, with the obvious exception of those metrics that explicitly depend on the total variation of in-flow. In particular, internalization regret is fairly stable, while the internalization rate tends to 100\% due to its definition~\eqref{eq:defInternalizationRate}.

As an ad-hoc adaptation, we propose to use the square-root of the quadratic variation to replace the total variation for use in Brownian model, as it remains stable over the number of jumps.

\begin{figure}[ht]
	\center
	\begin{subfigure}[b]{\textwidth}
		\includegraphics[width=\textwidth]{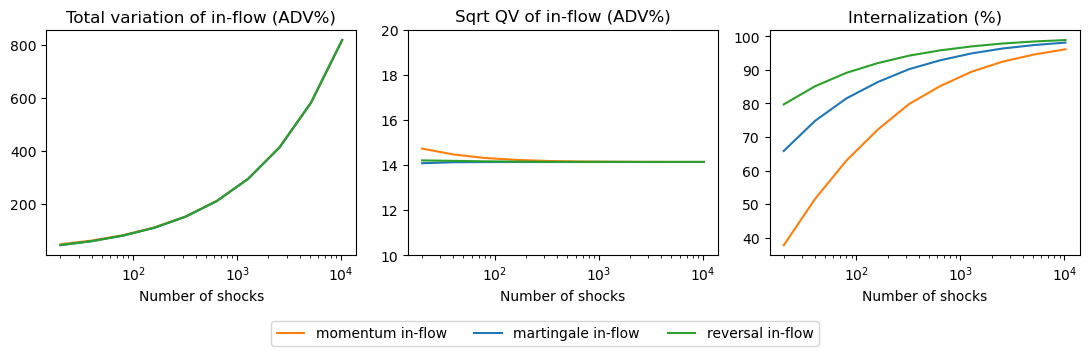}
		\caption{Total variation, quadratic variation, and internalization change in the passage to the diffusion limit.}
	\end{subfigure}
	\\
	\begin{subfigure}[b]{\textwidth}
		\includegraphics[width=\textwidth]{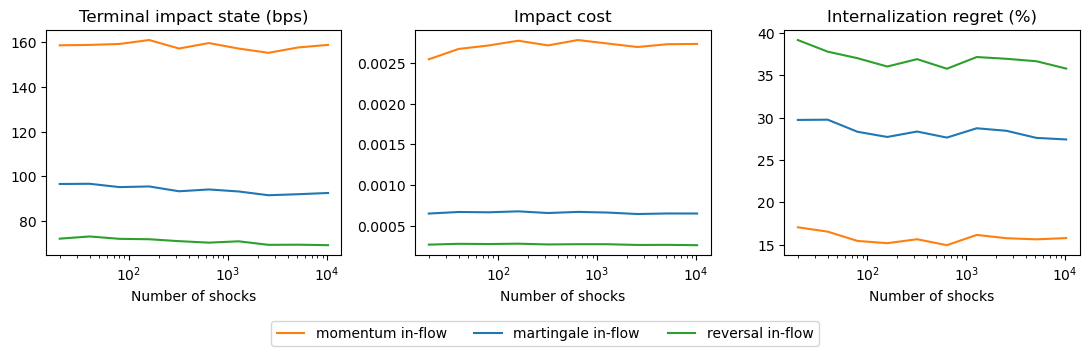}
		\caption{However, trading metrics that do not rely on total variation of in-flow remain stable.}
	\end{subfigure}
    \\
	\begin{subfigure}[b]{\textwidth}
		\includegraphics[width=\textwidth]{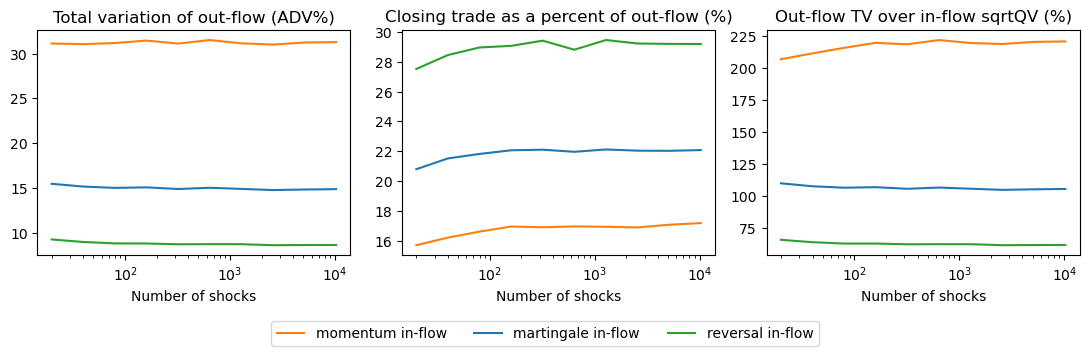}
		\caption{The square root of the quadratic variation of in-flow proxies well for the in-flow activity in the transition from jump process to It\^o process.}
	\end{subfigure}
        \caption{Sensitivity of average trading metrics to the diffusive limit for in-flow.}
	\label{fig:itoLimit}
\end{figure}

\FloatBarrier

\subsection{Sensitivity to the Initial Impact State}\label{se:initalCondition}

The optimal execution literature often assumes zero initial impact, and then shows that trade paths are monotone. Clearly paths need not to be monotone when the in-flow is stochastic, as the inventory can change sign. Beyond that, we emphasize in this section that the initial impact state matters (even in deterministic execution). This is straightforward but sometimes overlooked: The impact state forecasts reversal of the price, hence is equivalent to an alpha signal. Since the strategy maximizes P\&L, it balances this signal with the other costs of unwinding, which can lead to trades that \emph{add} to  the inventory rather than reducing it. Analytically, the formulas for $J_{0}$ and $q_{t}$ in Theorem~\ref{th:main} clearly show that trades can have either sign depending on the level of the impact state.

Trading based on past days' impact is far from an unwanted behavior in practice. Indeed, one challenge in trading systems is to correctly consider past days' trading in today's execution. For instance, \cite{WestrayRisk2021, Harvey2021, webster.23} tackle the issue from a transaction cost analysis (TCA) perspective. Unfortunately, the common solution of ignoring past days' trades leads to biases and sub-optimal trading. \cite{BouchaudBook} coins this problem ``issuer bias'' and highlights it as one of the four most common TCA and execution challenges among practitioners. While our solution is not unique in solving ``issuer bias'', it considers past days' trading with no additional effort within a consistent, tractable framework.

Figure~\ref{fig:impactPath} shows average trading paths from simulations with different initial impact states ($y= -30,\: 0,\: 30,\: 100$ bps) and martingale in-flow. A negative value is favorable in the sense that the price of a trade offsetting the initial in-flow $z>0$ is reduced. Indeed, in the favorable simulation, the unwind strategy hedges a sizable part of~$z$ in the opening auction. For increasing values of~$y$, the block trade is reduced; volume is deferred to exploit impact reversal. For high initial impact, reducing the impact state takes priority over reducing inventory and the opening trade adds to the initial inventory rather than hedging it. The strategy then trades aggressively during the day to unwind the acquired position.

\begin{figure}[htbp]
    \centering
    \includegraphics[width=1.0\textwidth]{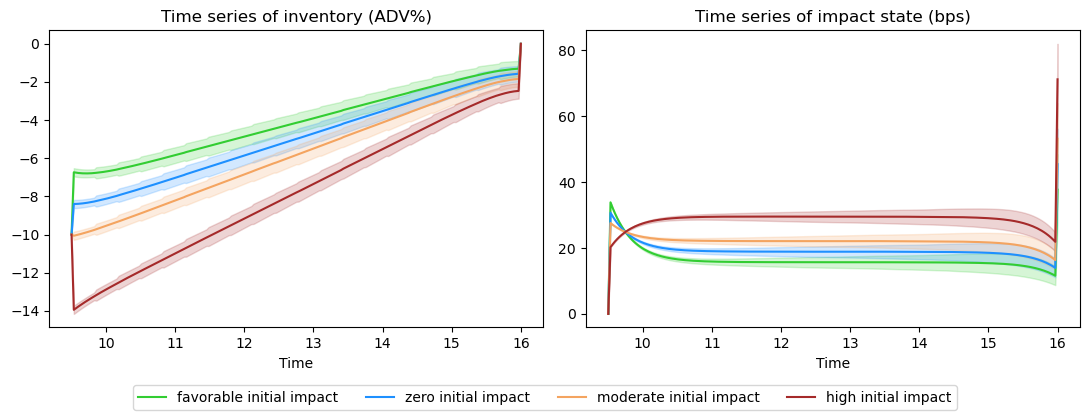}
    \caption{Sensitivity to initial impact state ($y$) for an initial inventory ($z$) and daily flow volatility ($\sigma$) of 10\% ADV: average metrics.}
    \label{fig:impactPath}
\end{figure}

\FloatBarrier

\subsection{A Representative Market-Maker Model for Stochastic Order flow}\label{se:empiricalAutocorr}

This section illustrates possible $\theta$ values from public market data. We preface this study with two important caveats to highlight that these estimates are purely presented for illustration purposes.

First, our estimates of market flow autocorrelation are simple, and there are many more sophisticated models and studies. Second, \emph{know your client.} Trading teams should calibrate stochastic models on their own in-flows and compare those to the representative market maker; they should not assume their in-flows behave like the aggregate market. For instance, competitive market makers may capture milder order flows, while less competitive desks may only capture toxic order flow, leading to more inventory risk. Similarly, a CRB's orders may not include smaller retail trades present on the tape.

The representative market-maker approach is a standard model for stochastic order flow in financial markets. For instance, market-making applications in \cite{CarteaBook} and the references therein assume that a trading strategy can capture a certain proportion of the overall market order flow. Empirical papers leverage the representative market-maker approach to provide statistics on market flow for trading applications. For example, \cite{Carmona2019} quantify market-maker's inventories and the follow-up paper \cite{Carmona2021} tracks individual traders' flows.
This section uses the representative market-maker approach to calibrate the proposed OU model for stochastic order flow on public market data. Therefore, the study discretizes the model for $Z$ using a step size $\Delta t$. Furthermore, $\sigma$ is normalized to one, as the representative market maker by definition captures 100\% of the day's trading volume:
\begin{align}
    Z_{t+\Delta t} = -\theta Z_{t} \Delta t + W_{t+\Delta t} - W_t.
\end{align}

Figure~\ref{fig:thetaDistribution} plots the \emph{doubling time} $log(2)/\theta$ as a function of the discretization choice $\Delta t$. The doubling time is to momentum flow what the half-life is to reverting flow. For each stock in the S\&P 500, the study estimates a single $\theta$ per stock over the year 2019, and the figure plots the corresponding distribution of doubling times across stocks. Our data source is LOBSTER \cite{lobsterIntro}, which provides all Nasdaq limit order book events in a standard table format and is aimed at academic research.

In line with \cite{BouchaudBook} and the references therein, we find that the market order flow consistently exhibits momentum across the S\&P 500 for time steps ranging from ten seconds to one hour. The corresponding doubling times are of the order of a couple of hours, to be compared to the roughly one-hour half-life of the OW model, as per \cite{JMK2022b}. Again, in line with \cite{Bouchaud2004}, order flow autocorrelation is roughly of the same magnitude as impact's decay, leading to approximately martingale impact state.

Furthermore, our study finds that $\theta$ depends on the discretization choice, in line with \cite{BouchaudBook, Toth2015} and the references therein, where order flow is found to have long-range autocorrelation and, therefore, cannot be calibrated on multiple time scales with a single OU model. Intuitively, different market participants, e.g., mutual funds, hedge funds, HFTs, submit orders over different timescales, e.g., minutes to days, and the aggregate order flow autocorrelation reflects each timescale. Furthermore, \cite{ButzOomen.19} empirically link client timescales to liquidity providers' internalization strategies.
\begin{figure}[ht]
\center
\includegraphics[width=0.5\textwidth]{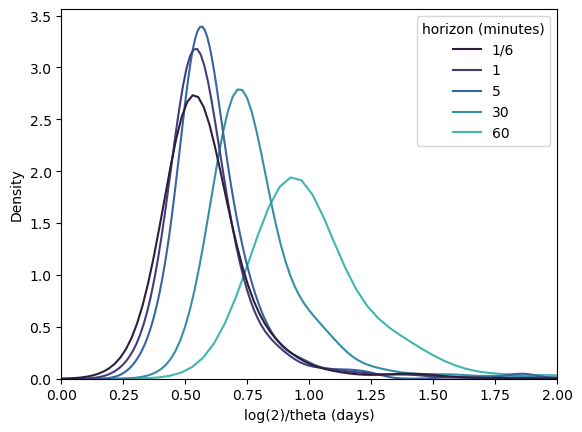}
\caption{Doubling time $log(2)/\theta$ as a function of the discretization choice $\Delta t$. The distribution covers the S\&P 500 for the year 2019.}\label{fig:thetaDistribution}
\end{figure}

\section{Proofs for Sections~\ref{se:optStrat}--\ref{se:jumps}: Dynamic Programming} \label{se:proofs}

\begin{proof}[Proof of Proposition~\ref{pr:costReformulated}]
    By its definition~\eqref{eq:execCost}, $\cC(J_{0}, q, J_{T})$ equals the sum of two expectations. First, $\E[\int_{0}^{T} S_{t}\,dQ_{t}]$, which equals $\E[S_{T}Z_{T}]$ by Lemma~\ref{le:costS}. And second,
    \begin{align*}
        \E \left[\int_{0}^{T}Y_{t}q_{t}\,dt 
        + \frac{1}{2} (Y_{0-}+Y_{0})J_{0} + \frac{1}{2} (Y_{T-}+Y_{T})J_{T} + \frac{1}{2} \int_{0}^{T}\eps_{t}q^{2}_{t}\,dt \right].
    \end{align*}
    Here Lemma~\ref{le:costY} shows that
    \begin{align}
        & \int_{0}^{T}Y_{t}q_{t}\,dt 
        + \frac{1}{2} (Y_{0-}+Y_{0})J_{0} + \frac{1}{2} (Y_{T-}+Y_{T})J_{T} \nonumber \\
        =& \int_{0}^{T} \frac{2 \beta_{t} + \dot\gamma_{t}}{2\lambda_{t}} Y^{2}_{t} \, dt 
        +  \frac{1}{2} \left( \frac{Y^{2}_{T}}{\lambda_{T}} - \frac{Y^{2}_{0}}{\lambda_{0}}
        + \frac{Y^{2}_{0}}{\lambda_{0-}} - \frac{y^{2}}{\lambda_{0-}} \right)\label{eq:costReformulated step1}.
    \end{align} 
    Moreover, using~\eqref{eq:impact jumps} and the liquidation constraint,
    \begin{align*}
       Y_{T} = Y_{T-} + \lambda_{T} J_{T} = Y_{T-} - \lambda_{T} X_{T} = Y_{T}^{c} - \lambda_{T}X_{T}.
    \end{align*}
    Substituting this into~\eqref{eq:costReformulated step1} and noting that $Y_{t}=Y^{c}_{t}$ for $t \in [0, T)$, the claim follows.
\end{proof}

\begin{lemma} \label{le:costS}
    Let $Q\equiv(J_{0}, q, J_{T}) \in \cA$. Then
    \begin{align} \label{eq:costS}
        \E \left[\int_{0}^{T} Q_{t}\, dS_{t} \right]=\E \left[ S_{0} J_{0} + \int_{0}^{T} S_{t} q_{t}\, dt + S_{T} J_{T} \right] = \E [S_{T} Z_{T}].
    \end{align}
\end{lemma}

\begin{proof}
    Note that $Q^{c}_{t} := \int_{0}^{t} q_{s} \, ds$ is a continuous process of finite variation. The integration-by-parts formula and $Q^{c}_{0} = 0$ then yield that
    \begin{align*}%
        \int_{0}^{T} S_{t} q_{t}\, dt = \int_{0}^{T} S_{t} \, dQ^{c}_{t} = S_{T} Q^{c}_{T} - \int_{0}^{T} Q^{c}_{t} \, dS_{t}.
    \end{align*}
    By the definition of admissibility, the last integral is a true martingale, hence has zero expectation. Admissibility also includes the liquidation constraint $J_{0} + Q^{c}_{T} + J_{T} = Z_{T}$ and thus
    \begin{align*}
        S_{0} J_{0} + S_{T} Q^{c}_{T} + S_{T} J_{T} &= S_{0} J_{0} + S_{T} (Z_{T} - J_{0}) 
        = -(S_{T} - S_{0})J_{0} + S_{T} Z_{T}. %
    \end{align*}
    Here $\E[(S_{T} - S_{0})J_{0}] = 0$ as $S$ is a martingale and $J_{0} \in \R$. The claim follows by combining the displays and taking expectation.
\end{proof}

\begin{lemma} \label{le:costY}
    Let $Q \equiv (J_{0}, q, J_{T}) \in \cA$. Then
    \begin{align} \label{eq:costY}
        \int_{0}^{T} Y_{t} q_{t} \, dt = \int_{0}^{T} \frac{2 \beta_{t} + \dot\gamma_{t}}{2\lambda_{t}} Y^{2}_{t} \, dt +  \frac{1}{2} \left(\frac{Y^{2}_{T-}}{\lambda_{T}}- \frac{Y^{2}_{0}}{\lambda_{0}} \right).
    \end{align}
    Moreover, 
    \begin{align} \label{eq:costY jumps}
        (Y_{0-} + Y_{0}) J_{0} = \frac{Y^{2}_{0}}{\lambda_{0-}} - \frac{y^{2}}{\lambda_{0-}} \qandq (Y_{T-} + Y_{T}) J_{T} = \frac{Y^{2}_{T}}{\lambda_{T}} - \frac{Y^{2}_{T-}}{\lambda_{T}}.
    \end{align}
\end{lemma}
\begin{proof}
      We first focus on $t\in (0,T)$. Note that $Y$ as defined in~\eqref{eq:impact} is continuous and has finite variation on $(0,T)$ when $Q$ is of the form \eqref{eq:Q main}. Using $\lambda_{t} = e^{\gamma_{t}}$ and $q_{t}\,dt = e^{-\gamma_{t}} (dY_{t} + \beta_{t} Y_{t} \,dt)$, we have
    \begin{align*}%
        \int_{0}^{T} Y_{t} q_{t} \, dt
        &= \int_{(0,T)}   e^{-\gamma_{t}} Y_{t} \, dY_{t} + \int_{0}^{T} e^{-\gamma_{t}} \beta_{t} Y^{2}_{t} \, dt.
    \end{align*}
    Applying integration-by-parts to $e^{-\gamma_{t}} Y^{2}_{t}$ over $t \in (0,T)$ yields that
    \begin{align*}%
        2\int_{(0,T)}   e^{-\gamma_{t}} Y_{t} \, dY_{t}
        &= e^{-\gamma_{T-}} Y^{2}_{T-} - e^{-\gamma_{0+}} Y^{2}_{0+} +  \int_{0}^{T} \dot\gamma_{t} e^{-\gamma_{t}} Y^{2}_{t} \, dt .
    \end{align*}
    Since $Y, \lambda$ are right-continuous at $t=0$ and $\lambda$ is also left-continuous at $t=T$, the claim~\eqref{eq:costY} follows by combining the displays. %
    The second claim, \eqref{eq:costY jumps}, follows by specializing
    \begin{align*}%
        (Y_{t} + Y_{t-})(Q_{t} - Q_{t-}) = \frac{1}{\lambda_{t-}}(Y_{t} + Y_{t-})(Y_{t} - Y_{t-})  = \frac{1}{\lambda_{t-}} (Y^{2}_{t} - Y^{2}_{t-})
    \end{align*}
    to $t\in\{0,T\}$ and recalling that $\lambda_{T-}=\lambda_{T}$.
\end{proof}

\begin{proof}[Proof of Proposition~\ref{pr:auxControl}]
    Omitting arguments for brevity, the generator of the controlled process is
\begin{align*}
    \cL^{q} v :&= (q + \theta z) \partial_{x} v + (-\beta y + \lambda q) \partial_{y} v - \theta z \, \partial_{z} v + \frac{\sigma^{2}}{2} \left(\partial_{xx} v - 2 \partial_{x z} v  + \partial_{zz} v \right) \\
    &= (\partial_{x} v + \lambda \partial_{y} v) q  -\beta \partial_{y} v \, y + \theta ( \partial_{x} v - \partial_{z} v)  z + \frac{\sigma^{2}}{2} \left(\partial_{xx} v - 2 \partial_{x z} v  + \partial_{zz} v \right).
\end{align*}
The running cost is 
$
    c(t, x, y, z, q) = \frac{2\beta_{t} + \dot\gamma_{t}}{2 \lambda_{t}} y^{2} + \frac{\eps_{t}}{2} q^{2},
$
so that the HJB equation $-\partial_{t} v - \inf_{q}( \cL^{q} v + c) = 0 $ reads
\begin{align} \label{eq:HJB}
    -\partial_{t} v - \inf_{q\in\R} & \left[ (\partial_{x} v + \lambda_{t} \partial_{y} v) q + \frac{\eps_{t}}{2} q^2 \right] + \beta_{t} \partial_{y} v \, y - \frac{2\beta_{t} + \dot\gamma_{t}}{2 \lambda_{t}} y^{2} \nonumber \\ %
    & - \theta_{t} ( \partial_{x} v - \partial_{z} v)  z - \frac{\sigma^{2}_{t}}{2} \left(\partial_{xx} v - 2 \partial_{x z} v  + \partial_{zz} v \right) = 0.
\end{align}
The candidate optimal control in feedback form is given by the minimizer in~\eqref{eq:HJB}, 
\begin{align*}
    q^{*}_{t} = - \eps^{-1}_{t} (\partial_{x} v + \lambda_{t} \partial_{y} v),
\end{align*}
which yields
\begin{align*}
    - \inf_{q\in\R} \left[ (\partial_{x} v + \lambda_{t} \partial_{y} v ) q + \frac{\eps_{t}}{2} q^2 \right] = \frac{1}{2\eps_{t}} ( \partial_{x} v + \lambda_{t} \partial_{y} v)^2.
\end{align*}
If $v\in C^{1,2}([0,T]\times\R^{3})$ is defined by~\eqref{eq:ansatz} with $A,\dots,F$ solving~\eqref{eq:riccati}, one verifies by direct calculation that~$v$ solves the HJB equation~\eqref{eq:HJB} as well as the terminal condition
\begin{align}\label{eq:terminalCondHJB}
        v(T, x, y, z) = \frac{1}{2}\left(\lambda_{T} x^2 - 2 x y +  \frac{y^2}{\lambda_{T}} \right)
\end{align}
corresponding to the terminal cost.
A standard verification argument then yields that $v$ is the value function and $q^{*}_{t}$ gives the optimal strategy.\footnote{To be precise, this requires~$q^{*}$ to be admissible. While $q^{*}$ is automatically square-integrable, we also need that $\int_0^t Q^{c}_s\,dS_s$ is a true martingale, where $Q^{c}_{t}:= \int_{0}^{t} q^{*}_{s} \, ds$. This is an integrability condition on~$S$, and is satisfied, e.g., if $dS_{t}=\nu_{t}\,dB_{t}$ where $B$ is a Brownian motion and $\nu$ is bounded.\label{fot:qAdmissible}} Uniqueness of the optimal control holds by strict convexity of~\eqref{valueFunDef2}.
\end{proof}

\begin{proof}[Proof of Proposition~\ref{pr:riccati}]
    We show that the Riccati system \eqref{eq:riccati} has a unique solution on $[0, T]$ and that $A_{t}, C_{t} \geq 0$ for $t \in [0, T]$.
   
   Consider first the (autonomous) ODE system for~$(A, B, C)$. It can be stated as the matrix Riccati equation
    \begin{align}
        \dot{\bP}_{t} = \bP_{t} \bB_{t} \bN_{t}^{-1} \bB_{t}' \bP_{t} - \bA_{t}' \bP_{t} - \bP_{t} \bA_{t} - \bC_{t}, \qquad t \in [0, T] \nonumber
    \end{align}
    with the terminal condition
    \begin{align*}
        \bP_{T} = \frac{1}{2} \begin{bmatrix} \lambda_{T} & - 1 \\ -1 & 1/\lambda_{T} \end{bmatrix},
    \end{align*}
    where 
    \begin{align*}
        \bP_{t} := \frac{1}{2} \begin{bmatrix} A_{t} & B_{t} \\ B_{t} & C_{t} \end{bmatrix}, \quad
        \bN_{t} := \frac{1}{2} \begin{bmatrix} \eps_{t} \end{bmatrix}
    \end{align*}
    \begin{align*}
        \bA_{t} := \begin{bmatrix} 0 & 0 \\ 0 & -\beta_{t} \end{bmatrix}, \quad
        \bB_{t} := \begin{bmatrix} 1 \\ \lambda_{t} \end{bmatrix}, \quad
        \bC_{t} := \begin{bmatrix} 0 & 0 \\ 0 & \frac{2\beta_{t} + \dot\gamma_{t}}{2\lambda_{t}} \end{bmatrix}.
    \end{align*}
    We note that $\bP_{T}$ is nonnegative definite. As the coefficients are bounded and $\eps_{t}$ was assumed to be bounded away from zero, \cite[Theorem 2.1(i)(ii)]{Wonham.68} shows that there exists a unique solution $\bP_{t}$ and that $\bP_{t}$ is nonnegative definite for all for $t\in [0,T]$. In particular, this yields $A_{t}\geq0$ and $C_{t}\geq0$.
    
   We can now consider $A, B, C$ as given bounded functions. The ODE system for~$(D, E, F, K)$ in \eqref{eq:riccati} is then linear with bounded coefficients, hence has a unique solution.
\end{proof}

\begin{proof}[Proof of Proposition~\ref{pr:riccatiAdditionalProperties}]  
(i) Nonnegative definiteness was established in the proof of Proposition~\ref{pr:riccati}.
We show that $C_{t} \leq \lambda^{-1}_{t}$ and $B_{t} + \lambda_{t} C_{t} > 0$ for $t \in [0,T)$. 
   
   Recall that the solution to a linear, first-order ODE $\dot y_{t} = p_{t} y_{t} - \ell_{t}$ with terminal condition $y_{T}$ is given by
    \begin{align} \label{eq:linearODE}
        y_{t} = e^{-\int_{t}^{T} p_{s} \, ds} y_{T} + \int_{t}^{T} e^{-\int_{t}^{s} p_{u} \, du} \ell_{s} \, ds.
    \end{align}
   In particular, $y_{t} \geq 0$ for $t \in [0, T)$ if $y_{T} \geq 0$ and $\ell_{t} \geq 0$ for $t \in [0, T)$, and $y_{t} > 0$ for $t \in [0, T)$ if additionally either $y_{T} > 0$ or $\ell_{t} > 0$ for $t$ in a non-degenerate interval of $[0, T)$. 

    We first prove that $C^{\dag}_{t} := -C_{t} + \lambda^{-1}_{t} \geq 0$ for $t \in [0, T]$. Indeed,
    \begin{align*}
        \dot C^{\dag}_{t} &= -\dot C_{t} - \lambda^{-1}_{t} \dot \gamma_{t} 
        = 2 \beta_{t} C^{\dag}_{t} - \eps^{-1}_{t} (B_{t} + \lambda_{t} C_{t})^{2}, \quad C^{\dag}_{T} = 0. %
    \end{align*}
    Seeing this as a special case of~\eqref{eq:linearODE} with $\ell_{t} := \eps^{-1}_{t} (B_{t} + \lambda_{t} C_{t})^{2} \geq 0$, we deduce $C^{\dag}_{t}\geq0$ for $t \in [0, T]$. 

    Lastly, define $M_{t}:= A_{t} + \lambda_{t} B_{t}$ and $N_{t} := B_{t} + \lambda_{t} C_{t} = B_{t} - \lambda_{t} C^{\dag}_{t} + 1$. Then
    \begin{align*}
        \dot N_{t} &= \dot B_{t} - \lambda_{t} (\dot C^{\dag}_{t} + \dot\gamma_{t} C^{\dag}_{t}) \\
        &= \eps^{-1}_{t} M_{t} N_{t} + \beta_{t} B_{t} + \lambda_{t} \eps^{-1}_{t}  N^{2}_{t} - 2 \lambda_{t} \beta_{t} C^{\dag}_{t} - \lambda_{t} \dot\gamma_{t} C^{\dag}_{t} \\
        &= (\eps^{-1}_{t} M_{t} + \lambda_{t} \eps^{-1}_{t}  N_{t} + \beta_{t} ) N_{t} -\lambda_{t} [(2\beta_{t}  +  \dot\gamma_{t})C^{\dag}_{t} + \beta_{t} C_{t}], \quad N_{T} = 0
    \end{align*}
    where we used $B_{t} = N_{t} - \lambda_{t} C_{t}$ from the second to the third line. Recall that $C_{t} \geq 0, \beta_{t} > 0$ and $2\beta_{t} + \dot \gamma_{t} > 0$.
    As $C_{T} = \lambda^{-1}_{T} > 0$, for $t$ close to $T$ we have $C_{t} > 0$ and thus $\ell_{t} := (2\beta_{t}  +  \dot\gamma_{t})C^{\dag}_{t} + \beta_{t} C_{t} > 0$.
    Using again \eqref{eq:linearODE}, we deduce $N_{t}>0$ for $t\in[0,T)$.
  
    (ii) Let $B^{\dag}_{t} := - B_{t}$, $M_{t}:= A_{t} + \lambda_{t} B_{t}$ and $N_{t} := B_{t} + \lambda_{t} C_{t}$.
    To show that $B_{t} < 0$ and $M_{t} > 0$ for $t \in [0, T)$, we first show that $(B^{\dag}_{t}, M_{t}) \in [0,\infty)^{2}$. Indeed, we have the system
    \begin{equation}
    \begin{cases}
        \dot B^{\dag}_{t} = \beta_{t} B^{\dag}_{t} - \eps^{-1}_{t} M_{t} N_{t}, & \quad B^{\dag}_{T} = 1 \\
        \dot M_{t} = \eps^{-1}_{t} ( M_{t} + \lambda_{t} N_{t} ) M_{t}  - \lambda_{t} (\beta_{t} + \dot\gamma_{t}) B^{\dag}_{t}, & \quad M_{T} = 0,
    \end{cases}
    \end{equation}
    which is of the general form $(\dot B^{\dag}_{t},\dot M_{t}) = \Phi(t, B^{\dag}_{t}, M_{t})$. Here the vector field  $-\Phi$ is inward-pointing on the boundary of the first quadrant $[0,\infty)^{2}$; that is, $-\Phi_{1} (t, b,m) \in [0,\infty)$ whenever $b=0,m\geq0$, and $-\Phi_{2} (t, b,m) \in [0,\infty)$ whenever $b\geq0,m=0$. 
    In view of the terminal condition $(B^{\dag}_{T},M_{T}) = (1,0) \in [0,\infty)^{2}$, it follows that $(B^{\dag}_{t},M_{t}) \in [0,\infty)^{2}$ for all $t\in[0,T]$. We can now apply~\eqref{eq:linearODE} together with $\beta_{t} + \dot \gamma_{t} > 0$ and $B^{\dag}_{T} = 1$ to deduce $(B^{\dag}_{t},M_{t}) \in (0,\infty)^{2}$ for all $t\in[0,T)$.

    Next, we set $V^{\dag}_{t} := M_{t} - (D_{t} + \lambda_{t} E_{t})$, $W^{\dag}_{t} := -B_{t} + E_{t}$ and consider
    \begin{align} \label{eq:ODE-VW}
        \begin{cases}
        \dot V^{\dag}_{t} = (\eps^{-1}_{t} M_{t} + \lambda_{t} \eps^{-1}_{t} N_{t} + \theta_{t}) V^{\dag}_{t} 
        - \lambda_{t} (\beta_{t} + \dot\gamma_{t}) W^{\dag}_{t}, & V^{\dag}_{T} = 0\\
        \dot W^{\dag}_{t} = (\theta_{t} + \beta_{t}) W^{\dag}_{t} - \eps^{-1}N_{t}V^{\dag}_{t}, & W^{\dag}_{T} = 1.
        \end{cases} 
    \end{align}
    Writing this as $(\dot V^{\dag}_{t},\dot W^{\dag}_{t})=\Phi(t, V^{\dag}_{t}, W^{\dag}_{t})$, we see that $-\Phi$ is inward-pointing on the boundary of the first quadrant. Repeating the argument above, we deduce that $(V^{\dag}_{t}, W^{\dag}_{t}) \in (0, \infty)^2$ for all $t \in [0, T)$. 
    
   (iii) Lastly, we establish the comparison result regarding $\theta$. Let $\tilde \theta:[0,T]\to\R$ is measurable and bounded, and define $\tilde D,\tilde E$ like $D,E$ but with $\tilde\theta$ instead of $\theta$. 
   Suppose that $\theta_{t} \geq \tilde \theta_{t}$ for $t \in [0, T]$, and
   define $(V^{\dag}_{t}, W^{\dag}_{t})$, $(\tilde V^{\dag}_{t}, \tilde W^{\dag}_{t})$ and $\Phi, \tilde \Phi$ as in \eqref{eq:ODE-VW}, where $(V^{\dag}_{T}, W^{\dag}_{T}) = (\tilde V^{\dag}_{T}, \tilde W^{\dag}_{T}) = (0,1)$.
   Check that
    \begin{align*}
        - \Phi_{1}(t, v, w) &\leq - \tilde \Phi_{1} (t, \tilde v, \tilde w), \quad \mbox{for} \quad 0 \leq v = \tilde v, \quad 0 \leq w \leq \tilde w \\
        - \Phi_{2}(t, v, w) &\leq - \tilde \Phi_{2} (t, \tilde v, \tilde w), \quad \mbox{for} \quad 0 \leq v \leq \tilde v, \quad 0 \leq  w = \tilde w.
    \end{align*}
    It follows that $V^{\dag}_{t} \leq \tilde V^{\dag}_{t}$ and $W^{\dag}_{t} \leq \tilde W^{\dag}_{t}$ for $t \in [0,T]$, or equivalently, 
    \begin{align*}
    D_{t} + \lambda_{t} E_{t} \geq \tilde D_{t} + \lambda_{t} \tilde E_{t} \qandq E_{t} \leq \tilde E_{t}, \quad t \in [0,T],
    \end{align*}
    as $M_{t}$ and $B_{t}$ do not depend on $\theta_{t}$ or $\tilde \theta_{t}$. In particular, we note that $\tilde D_{t}, \tilde E_{t} \equiv 0$ when $\tilde \theta_{t} \equiv 0$.
\end{proof}

\begin{proof}[Proof of Proposition~\ref{prop:optJ0}]
    We first note that the cost \eqref{eq:optCostJ0} is convex in $j_{0}$. Indeed, it follows from Proposition~\ref{pr:riccatiAdditionalProperties} that $v(t, x, y, z)$ is convex in $x$ and $y$, and thus $v(0, j_{0} - z, y + \lambda_{0-} j_{0}, z)$ is convex in $j_{0}$. The second term in \eqref{eq:optCostJ0} is also convex in $j_{0}$ as $0 < \lambda_{0-} \leq \lambda_{0}$.
    
    Since \eqref{eq:optCostJ0} is quadratic in $j_{0}$,  we can compute $\partial_{j_{0}} V(j_{0})$ where
    \begin{align} \label{eq:optJ0 step1}
        V(j_{0}):=v(0, j_{0} - z, y + \lambda_{0-} j_{0}, z) + \frac{1}{2} \left( \lambda^{-1}_{0-} - \lambda^{-1}_{0} \right) (y + \lambda_{0-} j_{0})^{2}.
    \end{align}
    Recall that
    \begin{align*}
        \partial_{x} v(0, x, y, z) = A_{0} x + B_{0} y + D_{0} z, \quad \partial_{y} v(0,x,y,z) =  B_{0} x + C_{0} y + E_{0} z,
    \end{align*}
    and
    \begin{align*}
        -\frac{f_{0-}}{\eps_{0}} = A_{0} + \lambda_{0-} B_{0}, \quad -\frac{g_{0-}}{\eps_{0}} = B_{0} + \lambda_{0-} C_{0}, \quad -\frac{h_{0-}}{\eps_{0}} = D_{0} + \lambda_{0-} E_{0}.
    \end{align*}
    Define $\eta_{0-} := -\eps^{-1}_{0} (1 - {\lambda_{0-}}/{\lambda_{0}}) \leq 0$. It is straightforward to find
    \begin{align*}
        \eps^{-1}_{0} \partial_{j_{0}} v(0, j_{0} - z, y + \lambda_{0-} j_{0}, z) 
        =&  -(f_{0-} + \lambda_{0-} g_{0-}) j_{0} - g_{0-} y + (f_{0-} - h_{0-}) z 
    \end{align*}
    and
    \begin{align}
        \eps_{0}^{-1}\partial_{j_{0}} V(j_{0}) = r j_{0} -  \left( g_{0-} + \eta_{0-} \right) y + (f_{0-} - h_{0-}) z \label{eq:optJ0again}
    \end{align}
    where $r:= -f_{0-} - \lambda_{0-} (g_{0-} + \eta_{0-})$.
    As $j_{0} \mapsto v(0, j_{0} - z, y + \lambda_{0-} j_{0}, z)$ is convex and quadratic, we must have $-f_{0-} - \lambda_{0-} g_{0-} \geq 0$.
    We show below that the inequality is strict, and thus $r > 0$ as $\eta_{0-} \leq 0$. It then follows that $V$ is not affine, hence strictly convex, and the desired formula for the optimal $j_{0}$ follows by setting \eqref{eq:optJ0again} to zero and solving for $j_{0}$.
    
    Suppose towards a contradiction that $-f_{0-} - \lambda_{0-} g_{0-}=0$. Recall that $f_{0-}$ and $g_{0-}$ do not depend on $y,z,\theta$, and consider the particular case where $y,z=0$ and $\theta,\sigma\equiv0$. The above calculation shows  $\partial_{j_{0}}v(0, j_{0},\lambda_{0-} j_{0}, 0)=0$ and in particular $v(0, j_{0},\lambda_{0-} j_{0}, 0)=v(0, 0,0, 0)=0$. In words, the round-trip strategy starting with an initial jump $j_{0}\neq0$ has zero cost, contradicting Corollary~\ref{co:noFreeRoundtrips}.

    Next, we show $g_{0-} + \eta_{0-} < 0$. Indeed, 
    \begin{align*}
        -\eps_{0} (g_{0-} + \eta_{0-}) &= B_{0} + \lambda_{0-} C_{0} + (1 - \lambda_{0-} / \lambda_{0})  \\
        &= B_{0} + \lambda_{0} C_{0} + (1 - \lambda_{0-} / \lambda_{0})(1-  \lambda_{0} C_{0}) > 0,
    \end{align*}
    where we recall from Proposition~\ref{pr:riccatiAdditionalProperties} that $1 - \lambda_{0} C_{0} \geq 0$ and $B_{0} + \lambda_{0} C_{0} > 0$.

    Lastly, suppose that $\beta_{t} + \dot \gamma_{t} > 0$ for all $t \in [0, T]$. Then
    \begin{align*}
        -\eps_{0}(f_{0-} - h_{0-}) &= (A_{0} + \lambda_{0-} B_{0}) - (D_{0} + \lambda_{0-} E_{0}) \\
        &= (A_{0} + \lambda_{0} B_{0}) - (D_{0} + \lambda_{0} E_{0}) - (\lambda_{0} - \lambda_{0-})(B_{0} - E_{0}) > 0,
    \end{align*}
    where we again recall from Proposition~\ref{pr:riccatiAdditionalProperties} that $A_{0} + \lambda_{0} B_{0} > D_{0} + \lambda_{0} E_{0}$ and $B_{0} < E_{0}$. 
\end{proof}

\begin{proof}[Sketch of proof for Proposition~\ref{pr:genMart}]
Given $q\in\cA$, introduce a process $V^{q}$ by $V^{q}_{t}=v(t,X,Y,Z_{t})$, where $v$ is defined in~\eqref{eq:ansatz} and $X,Y$ are the controlled states corresponding to~$q$. Note that $V^{q}_{0}=v(0,x,y,z)$ is deterministic and independent of $q$; it is the quantity~$V_{0}$ in the statement. Consider also the process
    $$
      \hat{V}^{q}_{t}:= \frac{1}{2} \int_0^t \left[ \frac{2\beta_{s} + \gamma_{s}'}{\lambda_{s}} (Y^{q})^{2}_{s} + \eps_{s} q^{2}_{s} \right] ds + V_{t}^{q}.
    $$
    In view of~\eqref{eq:terminalCondHJB}, the expected cost of $q\in\cA$ in~\eqref{valueFunDef} is given by $E [\hat{V}_{T}^{q}]$. A (lengthy) application of Ito's formula shows that the process $\hat{V}^{q}$ is a local submartingale for any $q\in\cA$, and a local martingale for $q=q^{*}$. This calculation uses the fact that~\eqref{eq:KgenMart} is equivalent to the generalized linear backward SDE
    \begin{equation*}
      dK_{t} = - \frac12 (A_{t} - 2 D_{t} + F_{t}) \,d[M,M]_{t} + dN_{t}, \qquad K_{T} = 0,
    \end{equation*}
    where $N$ is a martingale (namely, $N_{t}=\E[\int_{0}^{T}  \frac12 (A_{s} - 2 D_{s} + F_{s}) \,d[M,M]_{s}|\cF_{t}]$). The definition of~$\cA$ and an integrability condition on $M$ are chosen such that $q^{*}\in\cA$ and such that the local martingale terms are true martingales for any $q\in\cA$. Then, the (sub)martingale properties yield
    $$
      \E[\hat{V}_{T}^{q}]\geq \hat{V}_{0}^{q} = V_{0} = \hat{V}_{0}^{q^{*}}=\E [\hat{V}_{T}^{q^{*}}].
    $$
    This shows that $q^{*}$ is optimal and that $V_{0}$ is the value. Uniqueness of the optimal control is clear from the strict convexity of the cost.
\end{proof} 

\section{Proofs for Section~\ref{se:cstLiq}: Maximum Principle}\label{se:maxPrinciple}

Throughout this section, $\beta,\lambda,\eps\in (0,\infty)$ are constant; in particular, $\lambda_{0-} = \lambda$. 
The proof of Theorem~\ref{th:constParam} is based on the stochastic maximum principle. Consider the control problem~\eqref{valueFunDef}, i.e.,
    \begin{align}\label{valueFunDef2}
     \inf_{q \in \cQ} \frac{1}{2} \E \left\{  \int_t^T \left[ \frac{2\beta}{\lambda} (Y^{c}_{s})^{2} + \eps  q^{2}_{s} \right] ds 
     +\frac{1}{\lambda}(\lambda X_{T} - Y^{c}_{T})^{2} \right\}
\end{align}
with the state processes $(X_{s}, Y^{c}_{s}, Z_{s})_{s\in[t,T]}$ with initial time $t \in [0, T)$ and position $(x,y,z)$,
\begin{align} \label{eq:forward2}
    \begin{cases}
        dX_{s} = q_{s} \, ds - dZ_{s}, &\quad X_{t} = x \\
        dY^{c}_{s} = (-\beta Y^{c}_{s} + \lambda q_{s}) \, ds, &\quad Y^{c}_{t} = y \\
        dZ_{s} = -\theta_{s} Z_{s} \, ds + \sigma_{s} \, dW_{s}, &\quad Z_{t} = z
    \end{cases}
\end{align}
(For brevity, we write $Y$ instead of $Y^{c}$ throughout this section.)
We emphasize that $(t,x,y,z)$ should be seen as fixed parameters in the subsequent arguments; they are introduced in order to recover the full form of the feedback strategy in Theorem~\ref{th:main}. 

The maximum principle (e.g., \cite[Theorem~6.4.6]{Pham.09}) yields a BSDE system on $[t,T]$,
\begin{align} \label{eq:BSDE}
\begin{cases}
    -dU^{X}_{s} = - V^{X}_{s} \, dW_{s}, &\quad U^{X}_{T} = \lambda X_{T} - Y_{T} \\
    -dU^{Y}_{s} = (-\beta U^{Y}_{s} + \frac{2\beta}{\lambda} Y_{s})\, ds - V^{Y}_{s} \, dW_{s}, &\quad U^{Y}_{T} = - X_{T} + \lambda^{-1} Y_{T} \\
    -dU^{Z}_{s} = \theta_{s} (U^{X}_{s} - U^{Z}_{s})\, ds - V^{Z}_{s} \, dW_{s}, &\quad U^{Z}_{T} = 0,
\end{cases}
\end{align}
with optimality condition
\begin{align} \label{eq:max principle}
    U^{X}_{s} + \lambda U^{Y}_{s} + \eps q^{*}_{s} = 0, \quad s\in[t,T]
\end{align}
along the optimal state processes. Here $(q^{*}_{s})_{s\in[t,T]}$ denotes the optimal strategy in open-loop form (i.e., as a function of $\omega\in\Omega$).
 In particular, \eqref{eq:max principle} and the terminal conditions in~\eqref{eq:BSDE}  together yield that
\begin{align} \label{eq:opt T}
    q^{*}_{T} = 0.
\end{align}
We note that the maximum principle applies rigorously in our context, as we have already established that the value function is smooth (e.g., \cite[Theorem~6.4.7]{Pham.09}).

Recall from Proposition~\ref{pr:properties} that the desired functions $f_{t}, g_{t}, h_{t}$ do not depend on~$\sigma_{t}$. Thus, we may specialize to $\sigma\equiv0$ without loss of generality, rendering the in-flow process~$Z$ deterministic. The BSDE system~\eqref{eq:BSDE} then reduces to the backward ODE system
\begin{align} \label{eq:backward}
    \begin{cases}
        -\dot U^{X}_{s} = 0, &\quad U^{X}_{T} = \lambda X_{T} - Y_{T} \\
        -\dot U^{Y}_{s} = -\beta U^{Y}_{s} + \frac{2\beta}{\lambda} Y_{s}, &\quad U^{Y}_{T} = - X_{T} + \lambda^{-1} Y_{T} \\
        -\dot U^{Z}_{s} = \theta_{s} (U^{X}_{s} - U^{Z}_{s}), &\quad U^{Z}_{T} = 0.
    \end{cases}
\end{align}
(The above steps also apply in the case of time-varying coefficients, but we do not expect a closed-form solution in that situation.)
We recall that 
\begin{equation}\label{eq:tepsKappaDefn}
        \teps := \frac{\eps \beta}{2 \lambda}, \qquad \kappa := \beta \sqrt{1 + \teps^{-1}} > \beta,
\end{equation}
and note the relations
\begin{align} \label{eq:hint}
        \frac{1}{\kappa \pm \beta } = \frac{1}{\beta(\sqrt{1+\teps^{-1}} \pm 1)} = \frac{\teps}{\beta}\left( \sqrt{1+\teps^{-1}} \mp 1 \right) = \frac{\teps}{\beta}\left(\frac{\kappa}{\beta} \mp 1 \right)
\end{align}
which will be used repeatedly without further mention.

\begin{proposition} \label{prop:expexp}
    If $\beta,\lambda,\eps\in (0,\infty)$ are constant and $\sigma\equiv0$, the optimal trading speed (in open-loop form) for~\eqref{valueFunDef2}  is
    \begin{align} \label{eq:expexp}
        q^{*}_{s} := C^{1}_{t} e^{\kappa(T-s)} + C^{2}_{t}e^{-\kappa (T-s)} - (C^{1}_{t} + C^{2}_{t}), \quad s \in[t,T],
    \end{align}
    where $C^{1}_{t}$ and $C^{2}_{t}$ are 
    uniquely determined by the invertible linear system
    \begin{equation} \label{eq:lin sys}
        \bM_{t} \begin{bmatrix}
            C^{1}_{t} \\ C^{2}_{t}
        \end{bmatrix} = \begin{bmatrix}
            -\bar{x} \\ -\lambda^{-1} y
        \end{bmatrix}, \qquad
         \bM_{t} := \begin{bmatrix}
            a^{11}_{t} & a^{12}_{t} \\ a^{21}_{t} & a^{22}_{t}
        \end{bmatrix},
    \end{equation}
    where
    \begin{align*}
        \bar{x} &= x + (1-e^{-\int_{t}^{T} \theta_{u} \, du}) \, z \\
        a^{11}_{t} &=  \kappa^{-1}e^{\kappa (T-t)} -  (T-t) + [(\kappa+\beta)^{-1}-\beta^{-1} - \kappa^{-1}]  \\
        a^{12}_{t} &= -\kappa^{-1}e^{-\kappa (T-t)} - (T-t) + [-(\kappa-\beta)^{-1}-\beta^{-1} + \kappa^{-1}] \\
        a^{21}_{t} &= e^{\kappa  (T-t)} (\kappa - \beta)^{-1} + \beta^{-1}  \\
        a^{22}_{t} &= -e^{-\kappa  (T-t)}(\kappa +\beta)^{-1} + \beta^{-1}.
    \end{align*}
\end{proposition}
    
This result will be obtained by solving the forward-backward system consisting of the (forward) state dynamics~\eqref{eq:forward2} and the backward system~\eqref{eq:backward}, which are coupled through the optimality condition~\eqref{eq:max principle}.  

\begin{proof}
    \emph{Step 1:} We first derive a second-order, linear, non-homogeneous ODE for $q^{*}$.
    From the first equation in~\eqref{eq:backward}, $U^{X} \equiv U^{X}_{T}$ is constant, and hence the optimality condition~\eqref{eq:max principle} yields
\begin{equation*}
    U^{Y}_{s} = - \frac{\eps}{\lambda} q^{*}_{s} - \frac{1}{\lambda} U^{X}_{T}.
\end{equation*}
The second equation in~\eqref{eq:backward} then reads
\begin{equation*}
    \frac{\eps}{\lambda} \dot q^{*}_{s} = - \dot U^{Y}_{s} = \beta\left( \frac{\eps}{\lambda}  q^{*}_{s} + \frac{1}{\lambda} U^{X}_{T} \right) + \frac{2\beta}{\lambda} Y_{s},
\end{equation*}
or, equivalently,
\begin{equation} \label{eq:Y aux}
    Y_{s} = \frac{1}{2}\left( \frac{\eps}{\beta} \dot q^{*}_{s} - \eps q^{*}_{s} - U^{X}_{T} \right).
\end{equation}
The second equation in~\eqref{eq:forward2} now implies that
\begin{equation*}
   \frac{1}{2}\left( \frac{\eps}{\beta}  \ddot q^{*}_{s} - \eps \dot q^{*}_{s} \right) = - \frac{\beta}{2} \left(\frac{\eps}{\beta} \dot q^{*}_{s} - \eps q^{*}_{s} - U^{X}_{T} \right) + \lambda q^{*}_s
\end{equation*}
or 
\begin{equation} \label{eq:ODE v2}
    \ddot q^{*}_{s} - \beta^{2} \left( 1 + \frac{2\lambda}{\eps\beta} \right) q^{*}_{s} - \frac{\beta^{2}}{\eps} U^{X}_{T} = 0.
\end{equation}
We also have the terminal condition $q^{*}_T = 0$ from \eqref{eq:opt T}. With $\teps$ and $\kappa$ as defined in~\eqref{eq:tepsKappaDefn}, the solution to~\eqref{eq:ODE v2} can be written as
\begin{align} 
    q^{*}_s &= C^{1}_{t} e^{\kappa (T-s)} + C^{2}_{t} e^{-\kappa (T-s)} - \frac{\beta}{2\lambda (1+ \teps)} U^{X}_{T} \nonumber \\
   &= C^{1}_{t} e^{\kappa (T-s)} + C^{2}_{t} e^{-\kappa (T-s)} - (C^{1}_{t} + C^{2}_{t}) \label{eq:soln2 v2};
\end{align}
here 
\begin{align} \label{eq:cond1 v2}
    U^{X}_{T} = \frac{2\lambda}{\beta}(1+ \teps) (C^{1}_{t} + C^{2}_{t})
\end{align}
due to the terminal condition $q^{*}_T = 0$ and $C^{1}_{t}, C^{2}_{t} \in \R$ will be determined in Step~3 below.

\emph{Step 2:} We calculate the state processes $X, Y, Z$ in \eqref{eq:forward2} at $s \in [t, T]$. Note that
    \begin{align}
        Z_s = z e^{-\int_{t}^{s} \theta_{u} \, du}.
    \end{align}
    Then
    \begin{align}
        X_{s} &= x + \int_{t}^{s} q_{u} \, du - (Z_{s} - z) \nonumber \\
        &= x + C^{1}_{t} \left\{ \kappa^{-1}[e^{\kappa(T-t)} - e^{\kappa (T-s)} ] - (s-t) \right\} \nonumber \\
            & \qquad + C^{2}_{t} \left\{ \kappa^{-1}[e^{-\kappa (T-s)} - e^{-\kappa(T-t)}] - (s-t) \right\} \nonumber \\
            & \qquad + (1 - e^{-\int_{t}^{s} \theta_{u} \, du}) \, z \label{eq:Xs}
    \end{align}
    and
    \begin{align}
         Y_{s} &= e^{-\beta (s-t)} y + \lambda e^{-\beta s} \int_{t}^{s} e^{\beta u} q_u \, du \nonumber \\
         &= e^{-\beta (s-t)} y \nonumber \\
         &+ \lambda C^{1}_{t} \left\{ (\kappa-\beta)^{-1} [ e^{\kappa(T-t) - \beta (s - t)} - e^{\kappa (T-s)}] - \beta^{-1}[1 - e^{-\beta (s - t)}] \right\} \nonumber \\
         &+ \lambda C^{2}_{t} \left\{ (\kappa+\beta)^{-1}[e^{-\kappa (T-s)} - e^{-\kappa(T-t) - \beta (s - t)}] - \beta^{-1}[ 1 - e^{-\beta (s - t)}] \right\}. \label{eq:cev X2 v2}
    \end{align}
In particular, the state variables at~$T$ are
\begin{align}
    X_{T} &= x + C^{1}_{t} \left\{ \kappa^{-1}[ e^{\kappa (T-t)} - 1] -  (T-t) \right\} \nonumber \\
        & + C^{2}_{t} \left\{ \kappa^{-1}[1 - e^{-\kappa (T-t)}] - (T-t) \right\} +  (1 - e^{-\int_{t}^{T} \theta_{u} \, du}) \, z \label{eq:X1_T v2}
\end{align}
and
\begin{align}
    Y_{T} &= e^{-\beta (T-t)} y + \lambda C^{1}_{t} \left\{ (\kappa - \beta)^{-1} [e^{(\kappa-\beta)(T-t)} -1 ] -  \beta^{-1} [1-e^{-\beta(T-t)}]\right\} \nonumber \\
    & \qquad + \lambda C^{2}_{t} \left\{ (\kappa + \beta)^{-1} [1-e^{-(\kappa+\beta)(T-t)}] -  \beta^{-1} [1-e^{-\beta(T-t)}] \right\} \label{eq:X2_T v2}.
\end{align}

\emph{Step 3:} We use the expressions from Step 2 to determine $C^{1}_{t}$ and $C^{2}_{t}$. The derivation of \eqref{eq:Y aux} from the maximal principle \eqref{eq:max principle} shows that the forward-backward system is solved only if $Y$ satisfies 
\begin{align} \label{eq:Y aux2}
    Y_{s} = \frac{\teps \lambda}{\beta^{2}} \dot q^{*}_{s} - \frac{\teps \lambda}{\beta} q^{*}_{s} - \frac{1}{2} U^{X}_{T}, \quad s \in [t, T].
\end{align}
Using the expressions of $q^{*}_{s}$ in \eqref{eq:soln2 v2}, $U^{X}_{T}$ in \eqref{eq:cond1 v2} and $Y_{s}$ in \eqref{eq:cev X2 v2}, we claim that \eqref{eq:Y aux2} holds if and only if
\begin{align} \label{eq:lin1}
    & C^{1}_{t} \left[ e^{\kappa (T-t)} (\kappa - \beta)^{-1} + \beta^{-1} \right]  + C^{2}_{t} \left[ -e^{-\kappa (T-t)}(\kappa +\beta)^{-1} + \beta^{-1} \right] = -\lambda^{-1} y.
\end{align}
To wit, we can group the terms in \eqref{eq:Y aux2} by $e^{-\beta (s-t)}, e^{\kappa (T-s)}, e^{-\kappa (T-s)}$ and the remaining constants. The coefficient of $e^{-\beta (s-t)}$ being zero gives rise to \eqref{eq:lin1}; the other coefficients can be fully canceled thanks to \eqref{eq:hint}, so they do not yield further conditions.

On the other hand, the first terminal condition in \eqref{eq:backward} reads
\begin{equation} \label{eq:cond2 v2}
    -U^{X}_{T} + \lambda X_{T} - Y_{T} = 0.
\end{equation}
Note that \eqref{eq:Y aux2} at $T$ reduces to
\begin{equation} \label{eq:X2_T sim}
    Y_{T} = \frac{\lambda \teps \kappa}{\beta^2}(-C^{1}_{t} + C^{2}_{t}) - \frac{1}{2} U^{X}_{T},
\end{equation}
since $q^{*}_T = 0$.
Plugging \eqref{eq:cond1 v2}, \eqref{eq:X1_T v2}, \eqref{eq:X2_T sim} into \eqref{eq:cond2 v2}, we arrive at
\begin{align}
      & C^{1}_{t} \left[ (\kappa+\beta)^{-1}-\beta^{-1}+ \kappa^{-1}[ e^{\kappa (T-t)} - 1] - (T-t) \right] \nonumber \\
      + & C^{2}_{t} \left[ -(\kappa-\beta)^{-1}-\beta^{-1}+ \kappa^{-1}[1 - e^{-\kappa (T-t)}] - (T-t) \right] \nonumber \\
      & \qquad = -x - (1 - e^{-\int_{t}^{T} \theta_{u} \, du}) \, z \label{eq:lin2}
\end{align}
again with the help of \eqref{eq:hint}. 

The equations \eqref{eq:lin1} and \eqref{eq:lin2} forms the linear system \eqref{eq:lin sys}.
To calculate its determinant, note that
\begin{align*}
\det(\bM_{t}) &= a^{11}_{t}a^{22}_{t} - a^{12}_{t}a^{21}_{t} \\
&= -a^{11}_{t}e^{-\kappa  (T-t)}(\kappa +\beta)^{-1} - a^{12}_{t}e^{\kappa  (T-t)} (\kappa - \beta)^{-1} + \beta^{-1}(a^{11}_{t}-a^{12}_{t}),
\end{align*}
where thanks to \eqref{eq:hint}, we have
\begin{align*}
    a^{11}_{t} - a^{12}_{t} &= \kappa^{-1}( e^{\kappa (T-t)} + e^{-\kappa (T-t)}) + [(\kappa+\beta)^{-1} + (\kappa-\beta)^{-1}  - 2 \kappa^{-1}] \\
    &= \kappa^{-1}( e^{\kappa (T-t)} + e^{-\kappa (T-t)}) + 2 \left( {\teps \kappa}{\beta^{-2}} - {\kappa}^{-1} \right).
\end{align*}
We group the terms in $\det(\bM_{t})$ by $e^{\kappa (T-t)}$, $e^{-\kappa (T-t)}$, $(T-t) e^{\kappa (T-t)}$, $(T-t) e^{-\kappa (T-t)}$ and the constant terms. All coefficients are strictly greater than zero when $\teps, \beta, \lambda > 0$, and hence $\det(\bM_{t}) > 0$:
\begin{align*}
    \det(\bM_{t}) &= e^{\kappa (T-t)} \left\{ (\kappa - \beta)^{-1}[(\kappa-\beta)^{-1} + \beta^{-1} - \kappa^{-1}] + \beta^{-1} \kappa^{-1} \right\} \nonumber \\
    &+ e^{-\kappa (T-t)} \left\{ (\kappa +\beta)^{-1}[-(\kappa+\beta)^{-1} + \beta^{-1} + \kappa^{-1}] + \beta^{-1} \kappa^{-1} \right\} \nonumber \\
    &+ (T-t) e^{\kappa (T-t)} (\kappa - \beta)^{-1}
    + (T-t) e^{-\kappa (T-t)} (\kappa +\beta)^{-1}
    + 4 \teps \beta^{-1} \kappa^{-1}.
\end{align*}
In particular, the constant term is simplified from
\begin{align*}
    \kappa^{-1} & [-(\kappa +\beta)^{-1} + (\kappa - \beta)^{-1}] + 2\beta^{-1} \left( {\teps \kappa}{\beta^{-2}} - {\kappa}^{-1} \right)
    = 4 \teps \beta^{-1} \kappa^{-1}
\end{align*}
thanks to \eqref{eq:hint} and $\teps \kappa^2 \beta^{-2} = 1 + \teps $.
\end{proof}

We can now derive Theorem~\ref{th:constParam}.

\begin{proof}[Proof of Theorem~\ref{th:constParam}]
As mentioned above, by Proposition~\ref{pr:properties}, we can assume without loss of generality that $\sigma \equiv 0$. We can then apply Proposition~\ref{prop:expexp}.
    Rewriting the expression \eqref{eq:soln2 v2} at the initial time~$s=t$ yields
\begin{align}
    q^{*}_{t} &= C^{1}_{t}(e^{\kappa (T-t)} - 1) - C^{2}_{t}(1- e^{-\kappa (T-t)}) \nonumber \\
    &= (e^{\kappa (T-t)} - 1)(C^{1}_{t} - e^{-\kappa (T-t)} C^{2}_{t}). \label{eq:feedback0}
\end{align}
    From \eqref{eq:lin sys}, we have
    \begin{align}
        C^{1}_{t} &= -( a^{22}_{t} \bar x - \lambda^{-1} a^{12}_{t} y ) / d_{t}, \label{eq:C1} \\
        C^{2}_{t} &=  ( a^{21}_{t} \bar x - \lambda^{-1} a^{11}_{t} y ) / d_{t}, \label{eq:C2}
    \end{align}
	where $d_{t} := \det(\bM_{t})$.
    Note that
    \begin{align}
    C^{1}_{t} - e^{-\kappa (T-t)} C^{2}_{t} =( - \tilde{f}_{t} \bar x  + \tilde{g}_{t} y) / d_{t}, \label{eq:feedback1}
    \end{align}
    where
    \begin{align*}
        \tilde{f}_{t} :&= -(a^{22}_{t} + e^{-\kappa (T-t)} a^{21}_{t}) \\
        &= -[\beta^{-1} - (\kappa +\beta)^{-1}] e^{-\kappa  (T-t)} - [\beta^{-1} + (\kappa - \beta)^{-1}] < 0
    \end{align*}
    and
    \begin{align*}
        \tilde{g}_{t}:&= \lambda^{-1} a^{12}_{t} + \lambda^{-1} e^{-\kappa (T-t)} a^{11}_{t} \\
        &=  -  \lambda^{-1} (1 + e^{-\kappa (T-t)}) (T-t)  -  \lambda^{-1} [\beta^{-1} + 2\kappa^{-1} - (\kappa+\beta)^{-1}] e^{-\kappa (T-t)} \\ 
        &\quad -  \lambda^{-1} [(\kappa-\beta)^{-1}+\beta^{-1} -2\kappa^{-1}] < 0.
    \end{align*}
    In particular, it can be verified in the last line  that the constant term is negative,
    $$
        - \frac{1}{\beta}\left[ \frac{1}{\sqrt{1+\teps^{-1}} -1} + 1 - \frac{2}{\sqrt{1+\teps^{-1}}} \right] < 0.
    $$
    We can further rewrite $\tilde{g}_{t}$ as
    $$
    \tilde{g}_{t} =  \lambda^{-1} \tilde{f}_{t} -  \lambda^{-1} (1 + e^{-\kappa (T-t)}) (T-t) + 2  \lambda^{-1} \kappa^{-1} (1 - e^{-\kappa (T-t)}).
    $$
    Plugging \eqref{eq:feedback1} into \eqref{eq:feedback0} yields
    \begin{align*}
      q^{*}_{t} 
      &= (e^{\kappa (T-t)} - 1)( - \tilde{f}_{t} \bar x  + \tilde{g}_{t} y) / d_{t},
    \end{align*}
    completing the proof.
\end{proof}

\bibliographystyle{abbrv}
\bibliography{cite}

\end{document}